\documentclass{ieeeaccess}
\usepackage{cite}
\usepackage{amsmath,amssymb,amsfonts}
\usepackage{amsthm}
\usepackage{algorithm2e}
\usepackage{algorithmic}
\usepackage{graphicx}
\graphicspath{{img/}}
\usepackage{textcomp}

\usepackage{framed}
\definecolor{shadecolor}{rgb}{0.9,0.9,0.1}
\newcommand{\x}{}
\newcommand{\y}{}
\newcommand{\xx}{}
\newcommand{\yy}{}

\def\BibTeX{{\rm B\kern-.05em{\sc i\kern-.025em b}\kern-.08em
    T\kern-.1667em\lower.7ex\hbox{E}\kern-.125emX}}

\begin{document}
\history{Date of publication xxxx 00, 0000, date of current version xxxx 00, 0000.}
\doi{we.do.not.have.it.yet}

\newtheorem{theorem}{Theorem}[section]
\newtheorem{corollary}{Corollary}[]

\title{AoI-based Multicast Routing over Voronoi Overlays with Minimal Overhead}

\author{
\uppercase{Michele Albano}\authorrefmark{1}, \IEEEmembership{Senior Member, IEEE},
\uppercase{Matteo Mordacchini}\authorrefmark{2}, 
and
\uppercase{Laura Ricci}\authorrefmark{3}}
\address[1]{Department of Computer Science, Aalborg University, Denmark (e-mail: mialb@cs.aau.dk)}
\address[2]{Institute for Informatics and Telematics (IIT-CNR), Via Moruzzi 1, 56124, Pisa, Italy (e-mail: matteo.mordacchini@iit.cnr.it)}
\address[3]{Universit\'{a} di Pisa, Italy (e-mail: laura@di.unipi.it)}
\tfootnote{
The research has received funding from the EU ECSEL JU under the H2020 Framework Programme, JU grant nr. 737459 (Productive4.0 project).}

\markboth
{M. Albano, M. Mordacchini, L. Ricci: AoI-based Multicast Routing over Voronoi Overlays with minimal overhead}
{M. Albano, M. Mordacchini, L. Ricci: AoI-based Multicast Routing over Voronoi Overlays with minimal overhead}

\corresp{Corresponding author: Michele Albano  (e-mail: mialb@cs.aau.dk).}

\begin{abstract}
\x The increasing pervasive and ubiquitous presence of devices at the edge of the Internet is creating new scenarios for the emergence of novel services and applications. This is particularly true for location- and context-aware services.  These services call for new decentralized, self-organizing communication schemes that are able to face issues related to demanding resource consumption constraints, while ensuring efficient locality-based information dissemination and querying. Voronoi-based communication techniques are among the most widely used solutions in this field. However, when used for forwarding messages inside closed areas of the network (called Areas of Interest, AoIs), these solutions generally require a significant overhead in terms of redundant and/or unnecessary communications. This fact negatively impacts both the devices' resource consumption levels, as well as the network bandwidth usage. In order to eliminate all unnecessary communications,  in this paper  we present the MABRAVO (Multicast Algorithm for Broadcast and Routing over AoIs in Voronoi Overlays) protocol suite. MABRAVO allows to forward information within an AoI in a Voronoi network using only local information,  reaching all the devices in the area, and using the lowest possible number of messages, i.e., just one message for each node included in the AoI. The paper presents the mathematical and algorithmic descriptions of MABRAVO, as well as experimental findings of its performance, showing its ability to reduce communication costs to the strictly minimum required. \y
\end{abstract}


\titlepgskip=-15pt

\maketitle

\thanks{We thank Dr. Ranieri Baraglia for the invaluable support he gave  during the preliminary work that led to the results presented in this paper.}

\section{Introduction}\label{sec:intro}

We are witnessing a fast and vast expansion of the Internet at its edges~\cite{iopComcom}. This is mainly due to the pervasive diffusion in the environment of smart objects, like sensors, Internet of Things (IoT) devices, user personal devices, etc.

This scenario allows the emergence of novel services and applications~\cite{sayed18,sot16,wang17,iop2020human,delsing2017arrowhead,bellavista18,mordacchini2015crowdsourcing}, supported by potentially large networks of highly distributed and autonomous devices. \x Traditional centralized control and communication techniques \y do not suit the needs and requirements of such an environment.

In particular, these devices are usually equipped with computing and communication capabilities, that allow them to create and exchange information both among themselves and with other remote services. 
\x One of the most challenging problems is related to the fact that this kind of systems typically requires frequent exchanges of information among a large number of geographically dispersed devices. 
The communication complexity  is further increased by the fact that devices cannot always count on the support of central communication infrastructures, posing the need to apply autonomous, self-organizing forms of communication and interaction among devices~\cite{8890658,STRAYER20181,Alsharoa18,8379320}.  

Location- and context-aware services~\cite{Haosheng18} in this scenario are faced with additional issues. In fact, these services are characterized by the fact that most of the messages are directed (and of interest) only to limited/specific areas of the network. This is the case of communications directed to bounded regions of the space, like Areas of Interest~\cite{Liu2014}, validity regions~\cite{vrsense,mqry}, and safe regions~\cite{Qi18}.

As a consequence, effective and efficient communication schemes for this kind of application are of utmost relevance. This fact poses the challenge to devise information dissemination mechanisms that are able to face locality-based communication needs, while coping efficiently with demanding requests in terms of scalability, responsiveness, performance and resource consumption constraints.

\subsection{Contribution}\label{subsec:contribution}

In this paper, we focus on  Voronoi-based overlay networks for data communication and data dissemination among decentralized, autonomous entities. This is a widely used solution~\cite{ghaffari2010necessity}. \y In fact, Voronoi-based techniques have been presented as effective solutions for disseminating and querying data in decentralized, distributed systems. This kind of techniques have been successfully applied in the IoT~\cite{wan19}, wireless sensor networks~\cite{pietra19,CRPD18}, underwater networks~\cite{underwater17, underwater18}, embedded computing systems~\cite{taxi16}, vehicular networks~\cite{v2v18}, and even distributed virtual environments\cite{ghaffari2014dynamic,mordacchini2010hivory}. 

\x
More specifically, this paper presents a solution for information dissemination within bounded areas of the network. This communication paradigm can be relevant for a wide spectrum of applications for context-aware services at the Edge~\cite{abdellatif2019edge,an2018context,liao2019learning}. Following the literature on this subject, in this paper these bounded regions are called Areas of Interest (AoIs)~\cite{Liu2014}. Issues with decentralized communications toward nodes in a AoI are related to the fact that the entities in the system have to coordinate autonomously in order to determine the involvement of other nodes in the propagation  and delivery of the information, without relying on any form of centralized/global support. Geometric routing techniques are generally used in this kind of systems to deliver messages and queries towards interested areas. However, previous works highlight the risk for the system to incur in redundant messages (i.e., the same message is delivered to some nodes more than once) and/or unnecessary communications (i.e., messages are sent to nodes not related to the AoI)~\cite{VON, Voraque}. In order to overcome these issues, state-of-the-art solutions require that nodes in a Voronoi network should use additional data, such as the positions of the neighbors of a node's immediate neighboring nodes (e.g.~\cite{VON, Voraque, Buyukkaya2008}). All these facts compromise the efficiency of the system. In fact, all these communications (both redundant and unnecessary messages, and the ones needed for maintaining an updated neighbors-of-neighbors list) are expensive in terms of nodes' resource consumption and bandwidth usage. 

The contribution of this paper is to present a solution that is able to avoid all these costs by defining a novel decentralized communication scheme for AoIs that is able:
\begin{itemize}
\item to rely \emph{only} on strictly local information (i.e., the position of immediate neighbors \xx and their identifiers, which will be called IDs for short in the rest of the paper \yy );
\item to always deliver a message to \emph{all} the nodes in an AoI;
\item to totally avoid \emph{all} redundant communications;
\item to totally avoid \emph{all} unnecessary communications.
\end{itemize}

With our approach, the number of messages required to deliver data within an AoI is reduced to its minimum, thus saving nodes' battery and computational resources, as well as bandwidth usage. The proposed solution is based on geometric properties of Voronoi networks. At the best of our knowledge, this is the first technique that is able to achieve all these objectives.

In order to present our solution, in this paper we provide:
\begin{itemize}
\item a mathematical description of the proposed approach;
\item mathematical proofs of the correctness of the proposed solution;
\item an algorithmic description of the approach.
\end{itemize}

This paper does not deal with the decentralized maintenance of a Voronoi network, since several solutions are already available in the literature~\cite{Alsalih08,pietrabissa2016distributed,beaumont2007voronet} and can be used for this purpose. We do not deal either with the dynamic behavior of nodes (i.e., churning nodes). This paper presents the very first completely decentralized solution that allows to forward information within a delimited AoI using only local information, and achieving the lowest possible number of messages (i.e., just one message for each node included in the AoI). The purpose of this paper is to present such a solution, prove it is mathematically sound, and explain how to implement it. Discussing possible dynamic behaviors of the nodes would have added too much material to this paper, making the presentation of this work less coherent and less focused with respect to the main goal of the paper. We thus decided to leave the issues related to dynamic nodes to further research and future investigations.
\y

The rest of this paper is organized as follows:
 Section~\ref{sec:model} presents the model of the overlay network we consider.
Section~\ref{sec:theoalgo} defines the
 algorithms of the MABRAVO protocol suite, and proves that they are correct and computationally efficient.
Section~\ref{sec:sims} presents our simulation environment and the results it provided.
\x Section~\ref{sec:related} presents an overview of the literature about the topic of this paper. \y
Finally, Section~\ref{sec:conc} presents our conclusions about the topic at hand.

\section{Network Model}\label{sec:model}

Given a set of sites $S = {s_1 ... s_n}$ that are points in a plane, a {\em 2-dimensional Voronoi tessellation} is a partition of the plane into cells, which assigns to each site $s_i$ a cell $V_{s_i}$ that is the set of points closer to $s_i$ than to any other site $s_j \in S$, according to a given definition of distance.
In this paper we consider the classical Voronoi tessellation, which uses the $L^2$ metric as a distance:
\begin{displaymath}
||p_i,p_j|| = \sqrt{(x_i-x_j)^2+(y_i-y_j)^2}
\end{displaymath}
\noindent
where $(x_i,y_i)$ are the coordinates of the point $p_i$, and $(x_j,y_j)$ are the coordinates of the point $p_j$.
The cell $V_{s_i}$ associated to  the site $s_i = (x_i, y_i)$ is the locus of all the points in the plane that are closer to $s_i$ than to any other site, formally
\begin{eqnarray}\label{eq:cell}
p_k \in V_{s_i} \Leftrightarrow
\forall s_j:~ ||p_k, s_i|| \le ||p_k, s_j||
\end{eqnarray}

\Figure[t!](topskip=0pt, botskip=0pt, midskip=0pt)[width=3.0in]{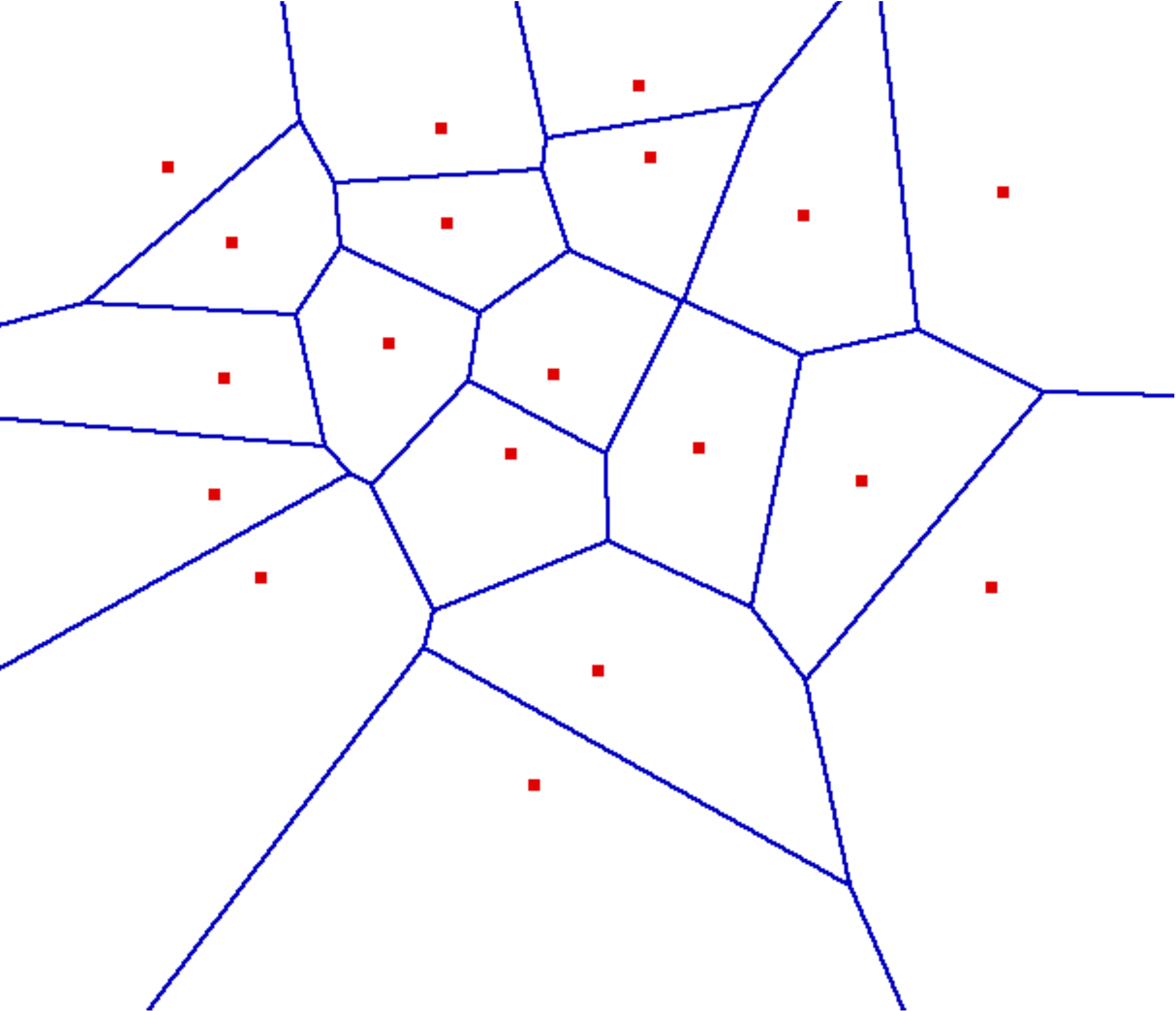}
{Classical Voronoi tessellation.\label{fig:voronoi_1}}

Figure~\ref{fig:voronoi_1} shows an example of a {\em Voronoi tessellation}.
Each inequality in Equation~\ref{eq:cell} is equivalent to dividing the plane into half-spaces, thus the cell is obtained by intersecting half-spaces, resulting in the cell being a convex polygon.
A cell $V_{s_i}$ may be characterized either by a finite area, or by an infinite area if some of the cell's sides are segments degenerated into half-lines.

A {\bf side} $l_{i,j}$ of a Voronoi cell is a segment that lays between two adjacent Voronoi cells $V_{s_i}$ and $V_{s_j}$, and a {\bf vertex} $v_{i,j,k}$ of a Voronoi cell is a point that is intersection between two sides of the Voronoi cell, and that lays between Voronoi cells $V_{s_i}$, $V_{s_j}$ and $V_{s_k}$. Should a segment degenerate into a half-line, the formalism still holds, except that the "segment" is adjacent to one vertex only.
A {\em Voronoi Overlay} is an overlay network that assigns the links among the sites following the Voronoi tessellations, i.e. a link exists in the overlay {\bf if and only if} the sites are Voronoi neighbors \cite{Beaumont07, VON, Voraque}.

In this paper, an \textbf{Area of Interest} (AoI) is a finite 2-dimensional convex region in the plane.
We consider a site $s_i$ to be in the AoI \textbf{if and only if} the intersection $I_{s_i}$ between its Voronoi cell $V_{s_i}$ and the AoI is not empty, formally $I_{s_i} = V_{s_i} \bigcap \textrm{AoI} \neq \emptyset$.

\begin{table}[t]
\caption{Definitions}\label{tab:def}
\begin{tabular}{|p{80pt}|p{145pt}|}
\hline
$D$ & either the destination of a routing message, or the originator of a AoI-cast request \\
\hline
AoI & convex area on the plane \\
\hline
$p_i = (x_i,y_i)$ & point $i$ of coordinates $x_i$ and $y_i$ \\
\hline
$s_i$ & site $i$ \\
\hline
$V_{s_i}$ & Voronoi cell of a site $s_i$ \\
\hline
$l_{i,j}$ & side between cells $V_{s_i}$ and $V_{s_j}$ \\
\hline
$v_{i,j,k}$ & vertex between cells $V_{s_i}$, $V_{s_j}$ and $V_{s_k}$ \\
\hline
$I_{s_i}$ & the intersection between $V_{s_i}$ and the AoI \\
\hline
$s_i$ is considered ``into'' the AoI & if $I_{s_i}$ is not empty \\
\hline
$Z_{s_i}(D)$ & union of the points of the segments that connect points of $I_{s_i}$ to $D$, with extremes not included, plus $D$ \\
\hline
$S_{s_i}(D)$ ({\em Segments of Interest} of $s_i$ towards $D$)
& intersection between $s_i$'s cell sides, and $Z_{s_i}(D)$ \\
\hline
$N_{s_i}(D)$ & neighbors of $s_i$ that share with $V_{s_i}$ a side with at least one point $\in S_{s_i}(D)$ \\
\hline
$V_{s_i,s_j}$ and $I_{s_i,s_j}$& $V_{s_i}$ and $I_{s_i}$ computed in the local vision of $s_j$\\
\hline
$Z_{s_i,s_j}(D)$, $S_{s_i,s_j}(D)$ and $N_{s_i,s_j}(D)$ & $Z_{s_i}(D)$, $S_{s_i}(D)$ and $N_{s_i}(D)$ computed in the local vision of $s_j$\\
\hline
\end{tabular}
\end{table}

Let us consider a point $D$, not necessarily co-located with a site.
We define $Z_{s_i}(D)$ as the union of the points of the segments that connect points of $I_{s_i}$ to $D$; $S_{s_i}(D)$ (the {\em Segments of Interest} of $s_i$ towards $D$) is defined as the intersection between $s_i$'s cell sides, and $Z_{s_i}(D)$; $N_{s_i}(D)$ is defined as the set of neighbors of $s_i$ whose sides share with $V_{s_i}$ at least one point $\in S_C(D)$. Finally, we define $V_{s_i,s_j}$, $I_{s_i,s_j}$, $Z_{s_i,s_j}(D)$, $S_{s_i,s_j}(D)$ and $N_{s_i,s_j}(D)$ in a manner analogous to $V_{s_i}$, $I_{s_i}$, $Z_{s_i}(D)$, $S_{s_i}(D)$ and $N_{s_i}(D)$, but computed with only the local information of $s_j$. Table~\ref{tab:def} summarizes the definitions presented in this section.

An \textbf{AoI-cast} is a routing protocol that delivers a packet to all the sites whose cells intersect the AoI.
The general strategy for an efficient AoI-cast is routing a packet from the sender to a site $s_i$ located into the AoI, and afterwards to create a distribution tree from $s_i$.
The lower limit for the number of required packets, corresponding to performing an AoI-cast over a tree, is equal to the number of sites in the AoI minus $1$. The rest of the paper considers that a packet has already reached one site in the AoI, and we are concerned with either routing it to another site in the AoI (unicast), or reaching all the sites in the AoI (AoI-cast).

\section{MABRAVO algorithms}\label{sec:theoalgo}

This section presents the MABRAVO protocol suite.
Recall from Section\ref{sec:model} that we define that site $s_i$ is part of the AoI {\bf if and only if} the intersection of its Voronoi cell $V_{s_i}$ and the AoI is not the empty set.

The MABRAVO routing algorithms consider the dichotomy between a ``global vision'' of the network and the ``local vision'' of a particular site, by defining the local vision of a site as the Voronoi tessellation computed by the site using only
the location data of its immediate neighbors.
Some previous works, for example VoRaQue~\cite{Voraque}, make use of non-local information such as knowledge regarding neighbors of neighbors. The maintenance of such information is prone to either a big communication overhead or data obsolescence. In fact, in order to have up-to-date information for the routing, each site needs to exchange a high number of messages with its neighbors, and this constitutes a burden on the overall performance of the system. Otherwise, there is an increased risk to incur in wrong forwarding decisions due to aged data.

In order to overcome these problems, we propose routing algorithms based uniquely on local information, i.e. a site knows only about its own location, and its Voronoi neighbors' locations \xx and IDs. \yy In doing this, we face issues like the ones presented in Fig.~\ref{fig:neighbors0}, and in Fig.~\ref{fig:aoi0}, which are mainly due to the discrepancies between one site's local vision of the network, and the real topology of the Voronoi diagram (global vision).

\begin{figure*}[!t]
\centering
\centering
\includegraphics[width=2.75in]{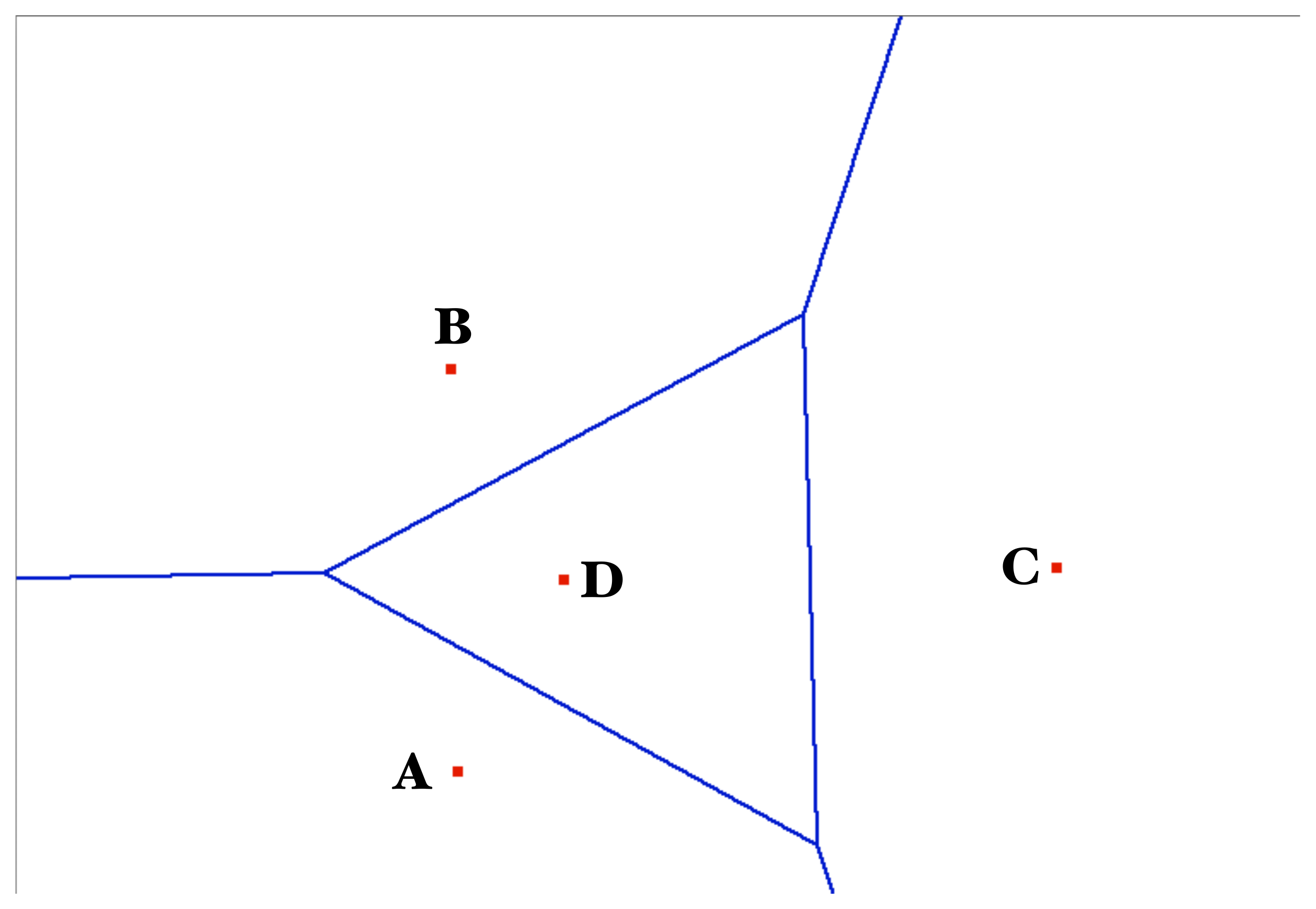}
\hfil
\centering
\includegraphics[width=2.75in]{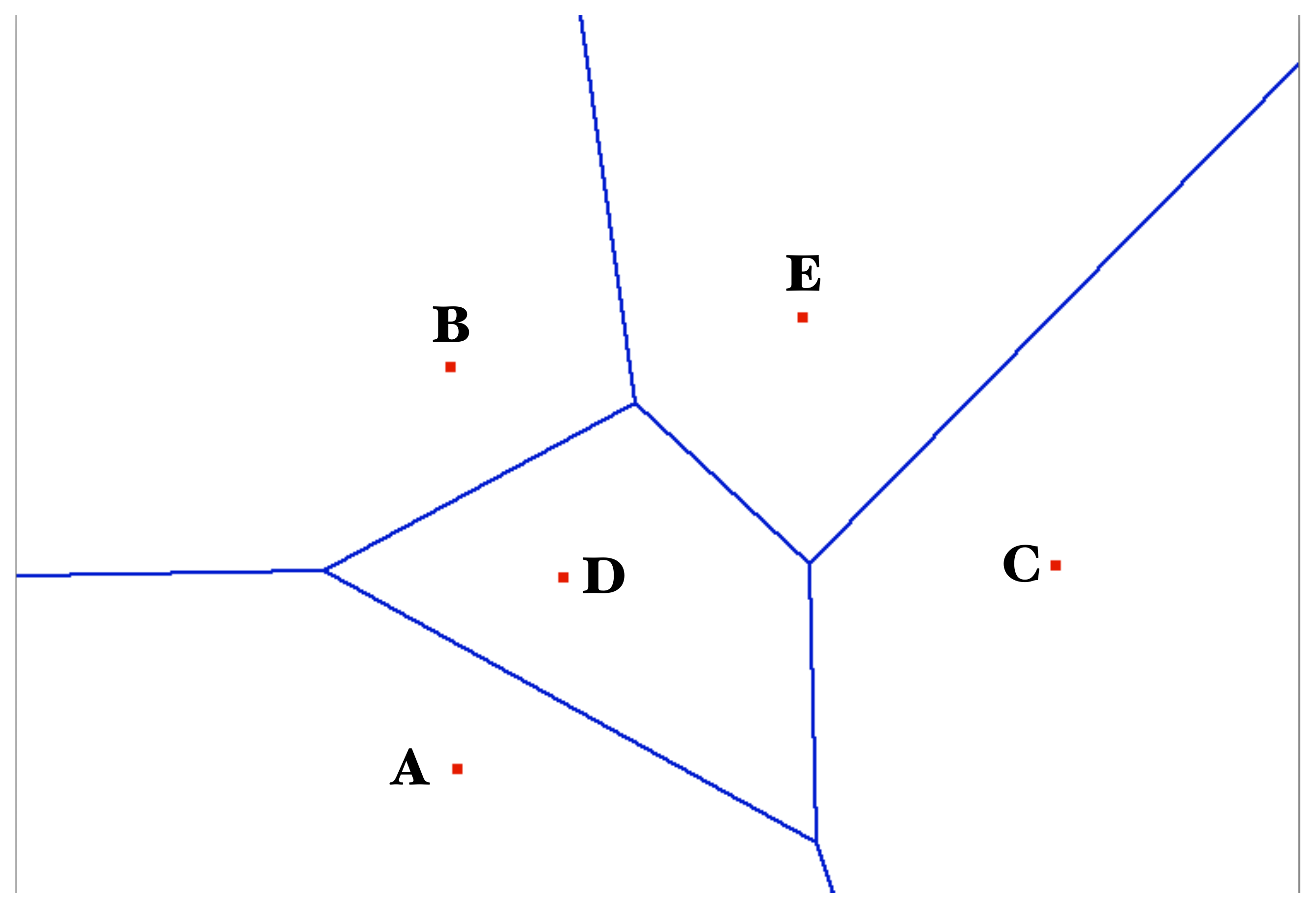}
\caption{Comparison between the local vision of site $A$ (left) and the real vision of the system (right). Site $A$ wrongly believes that $B$ and $C$ are neighbors }
\label{fig:neighbors0}
\end{figure*}

\begin{figure*}[!t]
\centering
\includegraphics[width=2.75in]{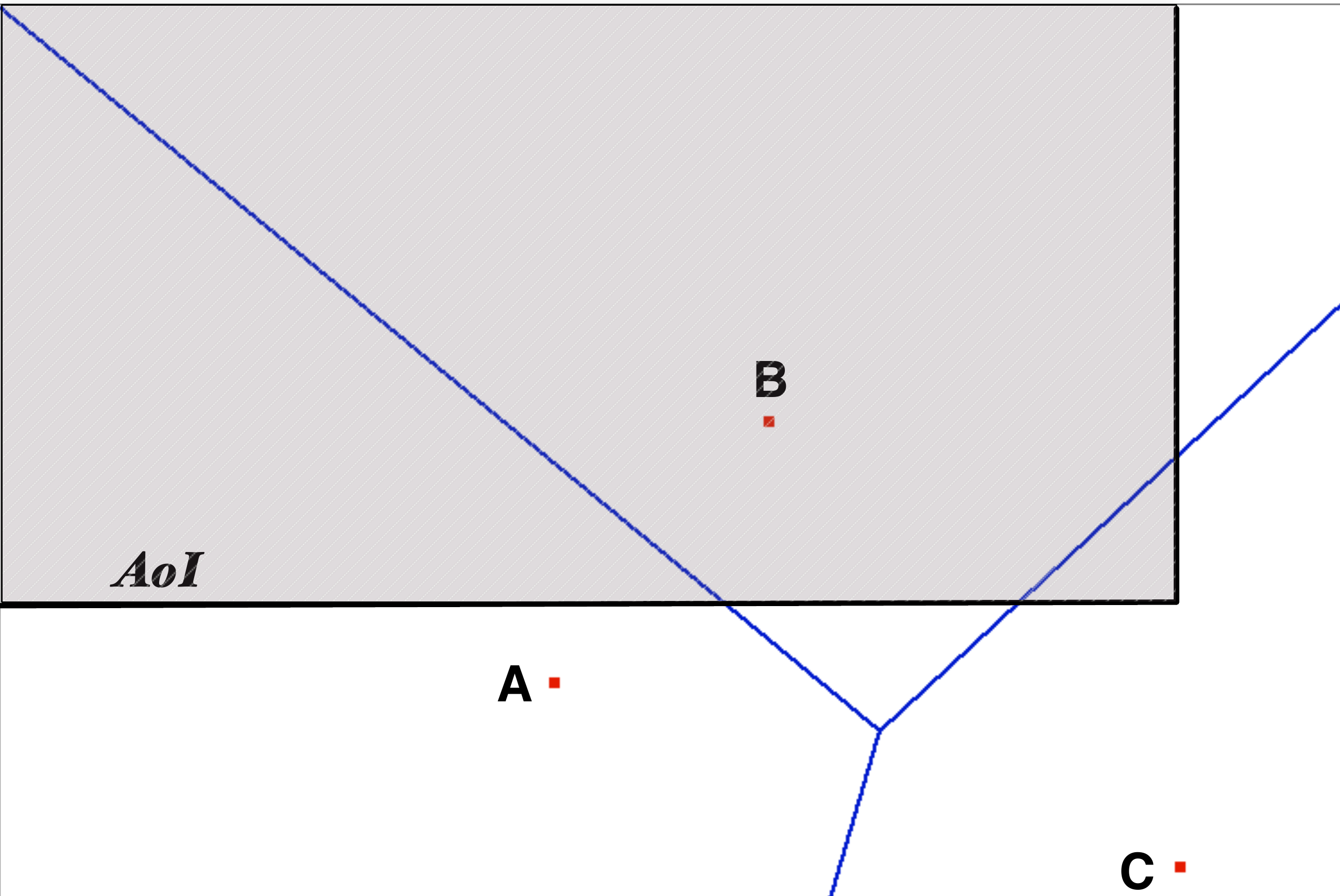}
\hfil
\centering
\includegraphics[width=2.75in]{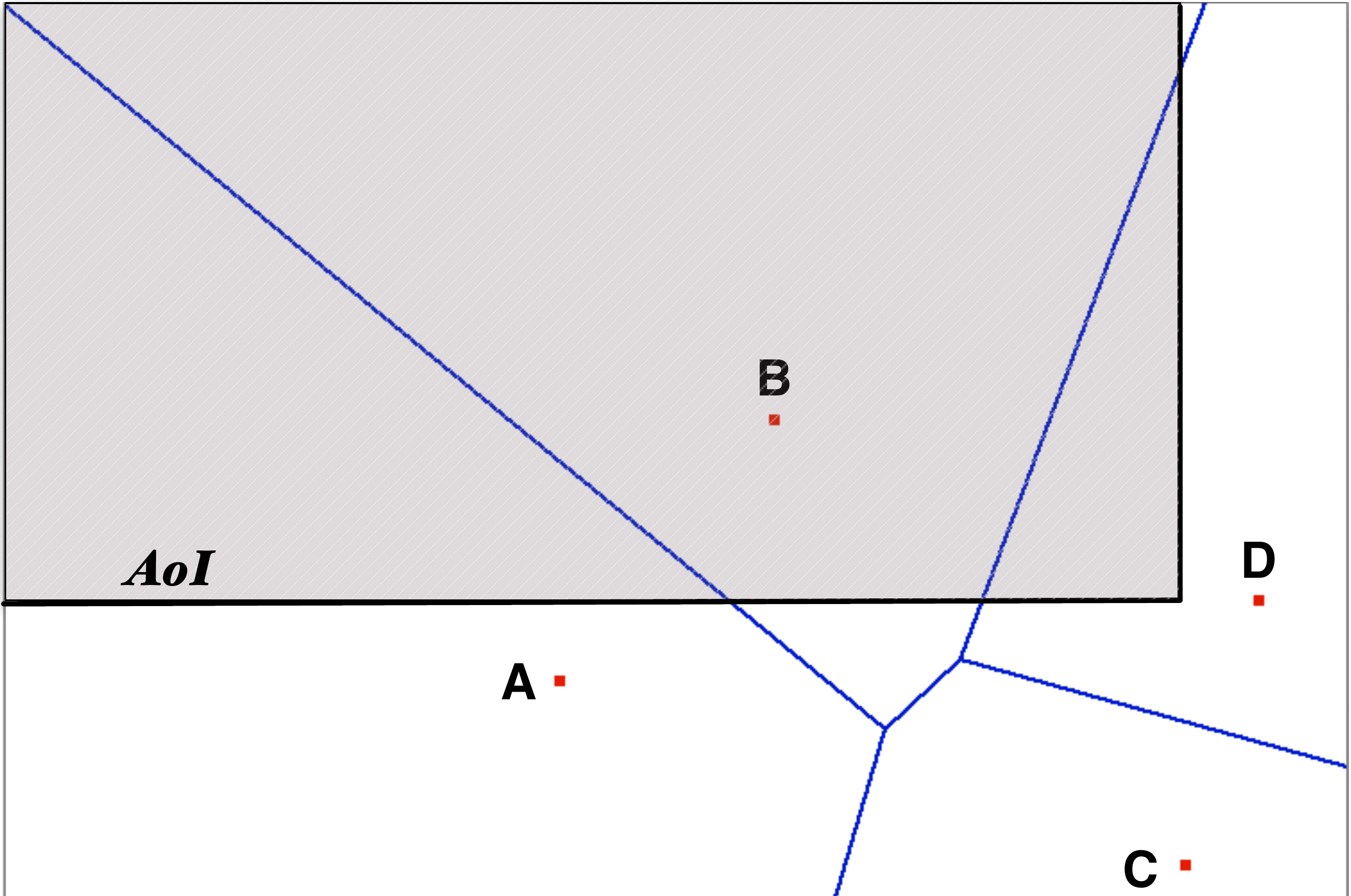}
\caption{Comparison between the local vision of site $A$ (left) and the real vision of the system (right). Site $A$ wrongly believes that $C$ is inside the $AoI$}
\label{fig:aoi0}
\end{figure*}

Left part of Fig.~\ref{fig:neighbors0} presents the local vision of site $A$, while right part of Fig.~\ref{fig:neighbors0} presents the global vision of the same area.
Let us suppose that an AoI-cast is being performed, for example using the routing protocol from~\cite{albanoCSE}, that $A$ received a packet to be delivered to $C$, and that $A$ believes that $B$ received the packet already.
$A$ may wrongly believe that sites $B$ and $C$ are mutual neighbors, and $A$ may consider that site $B$ is in charge of forwarding a packet it received to $C$, thus $A$ will not forward the packet to $C$ itself. In this particular case, it can happen that site $C$ will not receive the packet from any site, since $B$ and $C$ are not neighbors and $E$ could not lay in the AoI.

Another possible case is presented in Fig.~\ref{fig:aoi0}, where the left part of the figure is the local vision of $A$ and the right part represents the global vision.
In this case, $A$ could believe that its neighbor $C$ has a non-void intersection with the $AoI$.
Thus, $A$ decides to send the packet to $C$. The result is a useless message, since $C$'s Voronoi cell has no intersection with the $AoI$ and $C$ should not receive any message.

The purpose of the MABRAVO protocol suite is to overcome both the above problems, and to realize algorithms that, using the local vision of any site, are able to perform correct unicast and AoI-cast communication. Both communication modes avoid to contact unrelated sites (sites whose cells have no intersection with the AoI), and the AoI-cast uses the minimal number of packet transmissions, equal to the number of sites in the AoI minus 1.

The rest of this section describes the proposed MABRAVO protocol suite, starting with a discussion on the requirements that must be satisfied by the sites in terms of available primitives to allow an efficient implementation of unicast and AoI-cast routing (Subsection~\ref{subsec:req}), then presenting the unicast protocol MABRAVO$_D$ and proving its correctness (Subsection~\ref{subsec:theounic} and Subsection~\ref{subsec:theounicproof}), and then doing the same for the AoI-cast protocol MABRAVO$_R$ (Subsection~\ref{subsec:theomult} and Subsection~\ref{subsec:theomultproof}).

\subsection{Requirements for the algorithms}\label{subsec:req}
Recall from Section~\ref{sec:model} the definition of side and vertex of a Voronoi cell.
To be able to efficiently execute the algorithms, each site $s_i$ has to maintain a data structure with its Voronoi cell's vertices and sides.
For each vertex $v_{i,j,k}$, the data structure must be able to provide $s_j$ and $s_k$, which are the two Voronoi neighbors of $s_i$ that are at the same distance from $v_{i,j,k}$.
For each side $l_{i,j}$, the data structure must be able to provide $s_j$, which the neighbor of $s_i$ that is adjacent to $l_{i,j}$.
Moreover, given a neighbor $s_j$, the data structure must be able to provide the side $l_{i,j}$ that is adjacent to both $s_i$ and $s_j$, and the two Voronoi vertices $v_{i,j,k}$ and $v_{i,j,l}$ that are shared by the Voronoi cells of $s_i$ and $s_j$.

We propose to use circular lists for the neighbors, sides, and Voronoi vertices of $s_i$.
The data structures get updated whenever a new site is inserted or removed from the Voronoi diagram, and the cost of querying and updating the data structures of site $s_i$ is proportional to the number of neighbors of $s_i$. Since we are considering Voronoi diagrams in the plane, it has been proven that the mean number of neighbors of a site has an \x expected \y value lower than $6$ (see for example~\cite{Aurenhammer}) over large Voronoi diagrams. \x Thus, the expected cost for querying and updating the data structure used by the MABRAVO routing protocols is $O(1)$.

\xx We consider that each node is assigned a unique ID, which will be used to break ties in the MABRAVO algorithms. Various techniques are available for creating unique IDs in decentralised, distributed systems (e.g.~\cite{golodoniuc2017distributed, mahalle2020rethinking, kortesniemi2019improving}, just to name some recent examples). \yy

On a final note, we consider that the MABRAVO routing protocols can use a Time to Live (TTL) mechanism similar to the one of the AODV protocol~\cite{perkins2003ad}. In fact, the MABRAVO routing protocols are proved to work properly only when the topology is maintained timely, and the TTL mechanism protects the network in case the protocols are used in very dynamic networks before the topology gets maintained.
\y

\begin{figure*}[!t]
\centering
\centering
\includegraphics[width=2.75in]{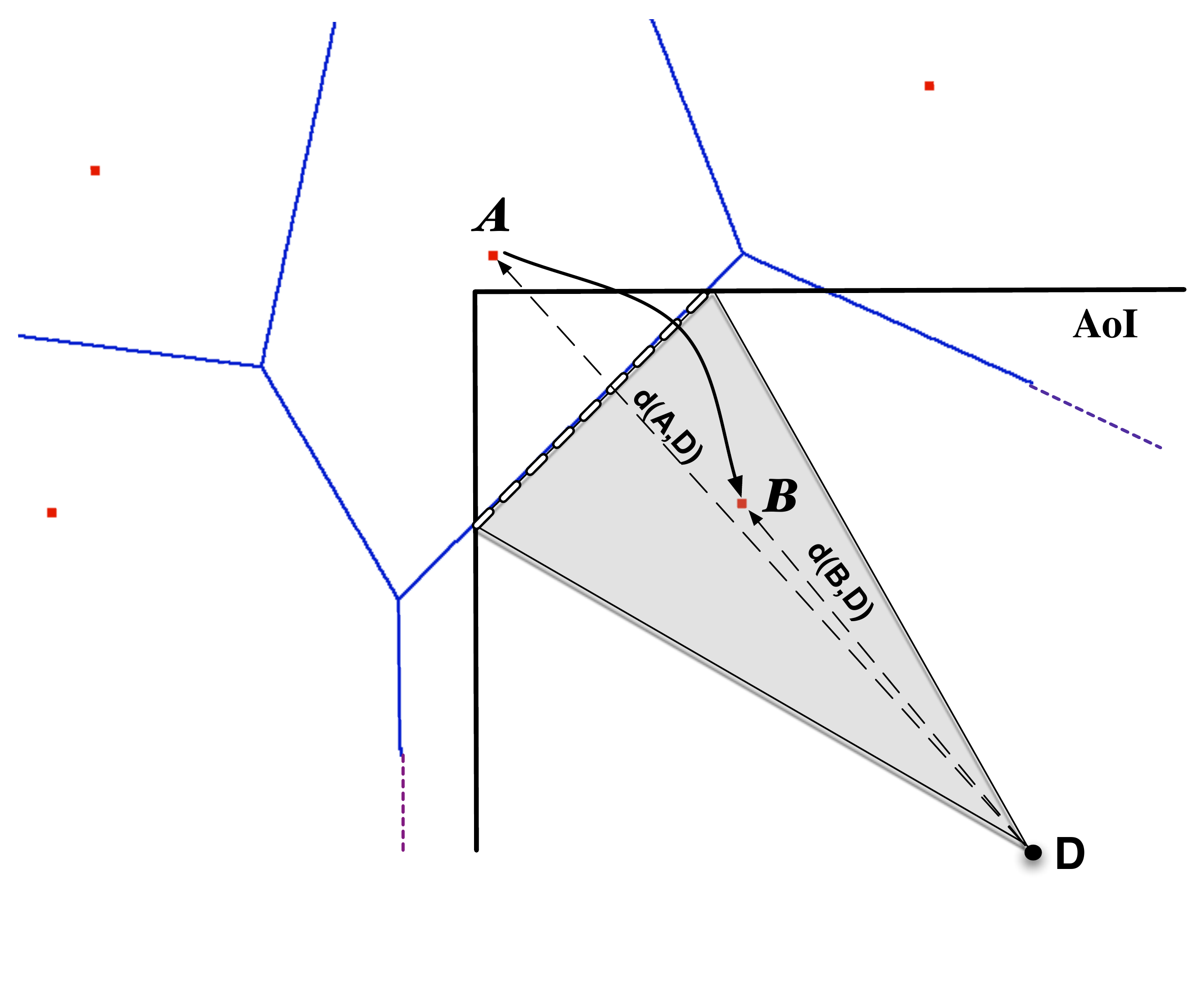}
\hfil
\centering
\includegraphics[width=2.75in]{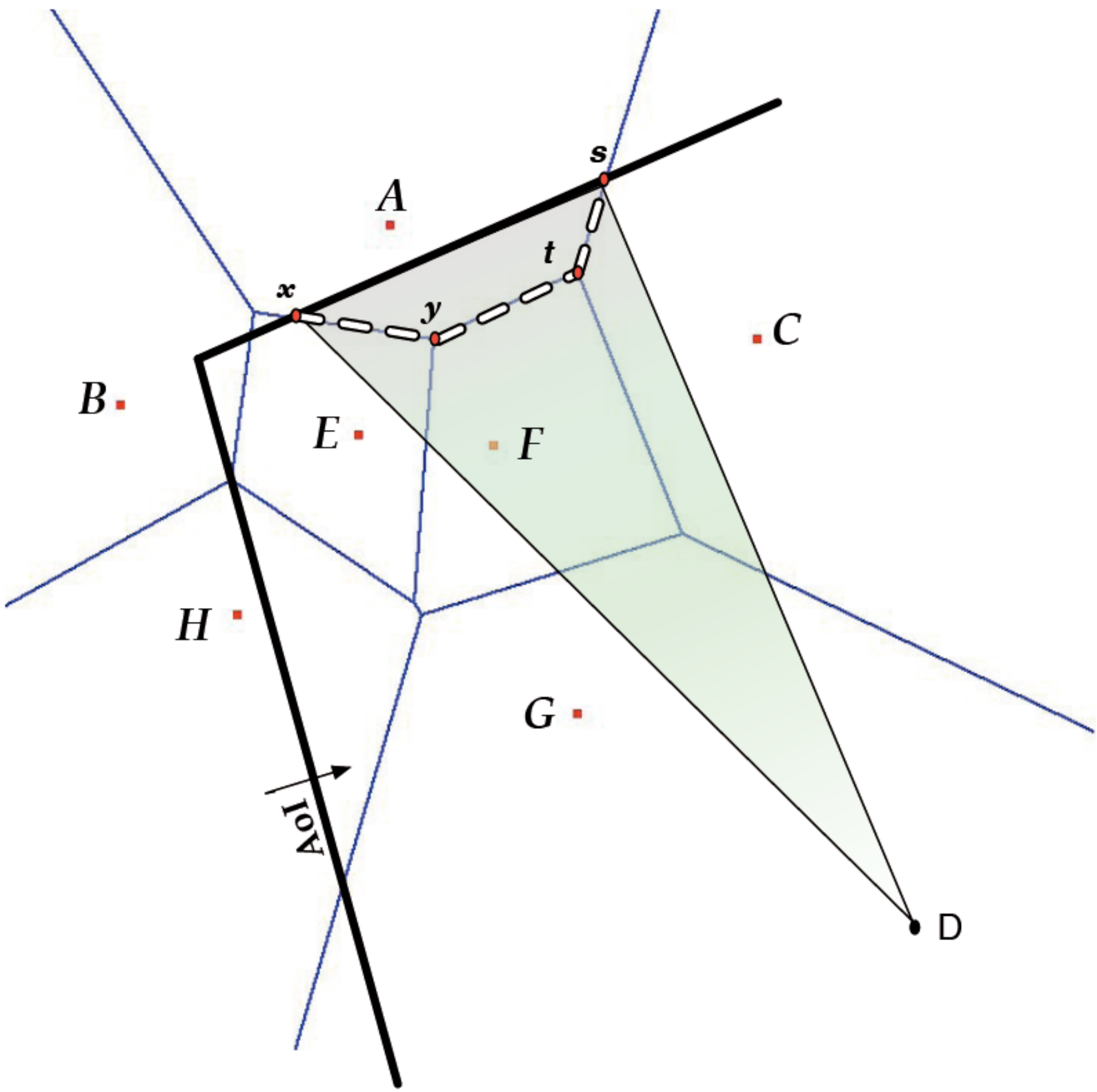}
\caption{Unicast routing step when $N_A(D)$ appears to be empty (left) and non-empty (right).}
\label{fig:unicastA}
\end{figure*}

\subsection{Unicast algorithm MABRAVO$_D$}\label{subsec:theounic}

Let us consider a site $C$ that is forwarding a packet towards the point $D$. $D$ can be co-located with a site or not; in the latter case, it must still belong to the AoI, and thus to the cell of a site belonging to the AoI. Site $C$ can be the initiator of the routing process, or an intermediate (and potentially final) hop.
Let us consider the actions performed by $C$ to decide which site to route the packet to and let us, with an abuse of notation, consider that $C$ is also the $C$'s index among the sites of the diagram (i.e.: $C = s_C$, sides of $V_C$ are $l_{C,i}$, and vertices of $V_C$ are $v_{C,i,j}$).

A unicast routing process will necessarily terminate with success if it respects the following properties:
\begin{itemize}
\item each site that has to forward the message shall be able to identify at least one other site to forward the message to (existence of next hop);
\item each site that has to forward the message shall uniquely identify the next forwarding site (unicity of next hop);
\item each routing step shall bring the packet closer (according to a given metric) to the destination point than the current site (implying the finiteness of the route).
\end{itemize}

If the current site $C$ is the closest to $D$ among its neighbors, the packet can be delivered and the routing process is completed. If not, it is necessary to route the packet to another site, and the quantities 
$Z_C(D)$, $S_C(D)$ and $N_C(D)$ (please refer to Section~\ref{sec:model} for their definitions) cannot be empty: $Z_C(D)$ contains at least one point $x$ since $C$ is in the AoI, $S_C(D)$ contains at least the intersection $y$ between one side of $V_C$ and the segment $xD$, $N_C(D)$ contains at least the neighbor of $C$ adjacent to the side containing $y$.

The unicast routing algorithm builds $N_C(D)$ by considering the vertices of $V_C$ laying in the AoI, and adding all neighbors that have a vertex in common with $C$ that is located in the AoI. This way of computing $N_C(D)$ can lead either to a non-empty set, or an empty set.
If $N_C(D)$ appears to be empty, it means that the AoI crosses a side twice (see left part of Fig.~\ref{fig:unicastA}), and there is only one potential next hop in $N_C(D)$, which will receive the packet.

If the $N_C(D)$ was not empty (see right part of Fig.~\ref{fig:unicastA}), the algorithm moves on the circular list of neighbors in $N_C(D)$ until it finds the one with lowest angle $\angle D C s_i$, and sends the packet to $s_i$; \xx in case of a tie, $C$ chooses the site with the lowest ID. \yy
For the latter case, an example is given in the right part of Fig.~\ref{fig:unicastA}. Site $A$ wants to send a packet towards the destination point $D$. Vertices $y$ and $t$ are in the AoI, thus the sites that will be considered as potential next hops are $E, F$ and $C$, since the three of them are also closer to $D$ than $A$. Site $B$ has a non-empty intersection with the $AoI$ in $A$'s local vision, but the side between $A$ and $B$ is not included in the \textit{segments of interest} and thus $B$ is discarded. Among the three sites, $F$ is selected since $DAF < DAE$ and $DAF < DAC$.

\begin{algorithm}[!t]
\caption{Algorithm MABRAVO$_D$, executed by $C$, having $D$ as destination}
\label{alg:impl:direct}
\DontPrintSemicolon
\SetAlgoLined
\If{ $\forall s_i$ neighbors of $C$, $D C < D s_i$}{
    Deliver the packet to $C$ \;
    \Return{}
}
Let $L = $ the set of all neighbors of $C$ \;
\ForEach{$s_i \in L$}{
    \lIf{ $D C < D s_i$}{
        Remove $s_i$ from $L$
    }\Else{
        \lIf{ ($v_{C,i,j} \notin AoI$) and ($v_{C,i,k} \notin AoI$)}{
            Remove $s_i$ from $L$
        }
    }
}
\eIf{$L = 0$}{
    \ForEach{ $s_i$ neighbor of $C$}{
        \lIf{ $D s_i < D C$}{
            \For{ all $q \in $ \textbf{sides} $(AoI)$}{
                \If{ $q \cap l_{i,C} \neq 0$}{
                    Send packet to $s_i$ \;
                    \Return{}
                }
            }
        }
    }
}{
    \ForEach{ $s_i \in L$}{
        Compute $a_i = \angle D C s_i$
    }
    Send packet to neighbor $s_{m}$ with lowest $a_{m}$, and lowest ID in case of a tie \; \Return{}
}
\end{algorithm}
The formal algorithm is reported in Algorithm~\ref{alg:impl:direct}, and next subsection proves that the MABRAVO$_D$ algorithm is correct.

\subsection{Proof of correctness for the unicast algorithm MABRAVO$_D$}\label{subsec:theounicproof}

\begin{theorem}
Given a site $C$ that is not the destination for a message, if $C$ receives a packet, there exists a site detected by the algorithm which is the next destination of the packet.\label{th:uni:exi}\end{theorem}
\begin{proof}
Take a point $P$ in $I_C$, which cannot be the empty set since $C \in AoI$.
Connect $P$ to $D$ with a straight segment, and consider the segment's intersection with $C$'s Voronoi sides. Considering the definition of Voronoi cell $V_C$ as given by Equation~\ref{eq:cell}, we can have two possible cases:
\begin{itemize}
\item
if there is no intersection, $PD$ lies in $C$'s Voronoi cell $V_C$. In this case, $D \in V_C(D)$, site $C$ owns point $D$, and the routing process is completed with success;
\item
otherwise, notice that the segment lays in $I_C$ since both $V_C$ and the AoI are convex.
Consider the intersection $q$ between $\overline{PD}$ and $V_C$'s borders. $q \in S_C(D)$, hence $N_C(D)$ comprises at least the site on the other side of $q$, which is closer to $D$ than $C$.
Thus, there exists at least one possible site to forward the message to.
\end{itemize}
\end{proof}
\begin{theorem}The fact that MABRAVO$_D$ forwards the packet from a site $C$ to a site $B$, implies that $B$ is closer to $D$ than $C$.\label{theo:distances}\end{theorem}
\begin{proof}
The chosen site $B$ is in $N_C(D)$. From definition of $N_C(D)$, at least one of the segments connecting $I_C$ to $D$ crosses the border between $C$ and $B$. Hence, using the definition of a cell's borders presented in Eq.~\ref{eq:cell}, since $D$ resides on the other side of the border between $B$ and $C$, we have proved that $BD < CD$.
\end{proof}
\begin{theorem}Unicast route is unique, and finite.\label{th:uni:uni}\end{theorem}
\begin{proof}
Since the algorithm for unicast routing is deterministic (it has no random component), at each step it can choose only one site as the next hop to go towards the routing destination D. Hence, the routing path is uniquely defined by the routing algorithm.
Each routing step brings the packet to a site that is closer to $D$ than the preceding site. Thus, a route can have at most as many hops as the number of sites in the network. Thus, the route is finite.
\end{proof}

\subsection{AoI-cast algorithm MABRAVO$_R$}\label{subsec:theomult}

This subsection presents MABRAVO Reverse (MABRAVO$_R$), which is an AoI-cast protocol that builds over the results presented in subsection~\ref{subsec:theounicproof} to compute AoI-cast trees in a distributed manner with local information only.
The rationale is that the algorithm MABRAVO$_R$, formalized in Algorithm~\ref{alg:impl:inv}, understands if MABRAVO$_D$ would route a packet from $C$ to $D$, and in that case $D$ sends the packet to $C$ while executing MABRAVO$_R$. The algorithm performs correct routing and minimizes the number of exchanged messages, by delivering
\begin{itemize}
\item
one message - and one message only - to each site whose Voronoi cell has a non-void intersection with the AoI, and
\item
no messages to sites outside the AoI.
\end{itemize}

\begin{figure*}[t]
\begin{center}
\includegraphics[width=0.48\linewidth,height=0.44\linewidth]{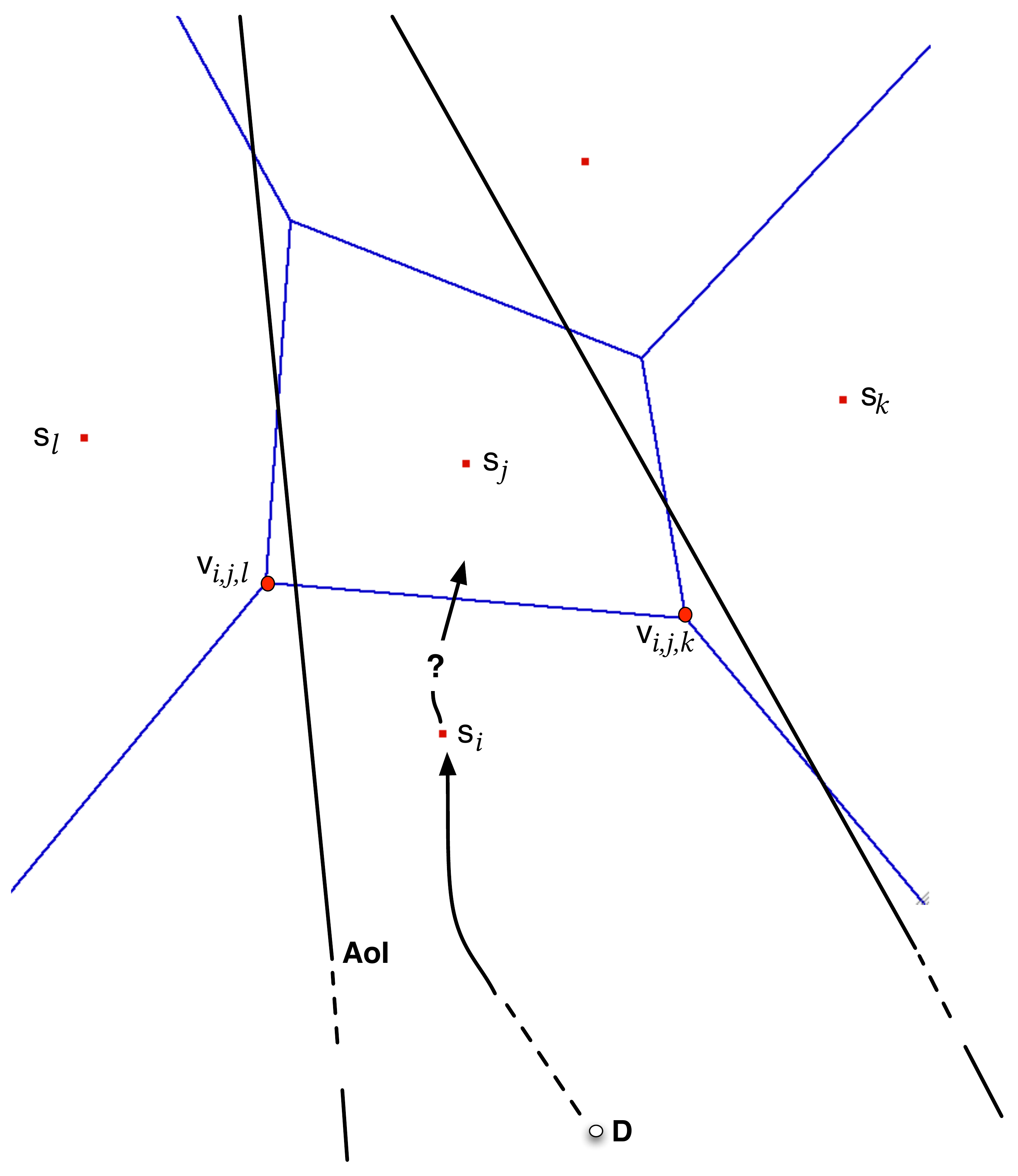}
\caption{Example of MABRAVO$_R$ routing.\label{fig:impltree}}
\end{center}
\end{figure*}

Let us start the presentation with an example regarding the execution of MABRAVO$_R$ on \x Fig.~\ref{fig:impltree}, \y where a site $s_i$ verifies if it should send the packet originated in $D$ to site $s_j$. The algorithm will do that if site $s_j$ would send a packet to $s_i$ to reach the destination $D$ when using the MABRAVO$_D$.
First of all, if $s_i D > s_j D$, $s_j$ can not be child of $s_i$ in the AoI-cast tree.
Let us now call $s_k$ and $s_l$ the two sites that are common neighbors of $s_i$ and $s_j$
(it is possible that one of the sites or both do not exist). Let also be
$v_{i,j,k}$ the Voronoi vertex adjacent to $s_i, s_j$ and $s_k$ and let be $v_{i,j,l}$ the Voronoi vertex adjacent to $s_i, s_j$ and $s_l$.
Algorithm~\ref{alg:impl:inv} considers two main cases:
\begin{itemize}
\item
Neither $v_{i,j,k}$ nor $v_{i,j,l}$ are in the AoI. In this case, $s_i$ checks
 if its border with $s_j$ crosses the AoI boundaries, in line with the routing performed in Fig.\ref{fig:unicastA}.
If this is true, $s_i$ is the only feasible next hop of $s_j$ in MABRAVO$_D$, thus
 $s_i$ sends the packet to $s_j$;
\item
$v_{i,j,k}$ or $v_{i,j,l}$ or both lay into the AoI. In this case, $s_k$ or $s_l$ or both sites are compared with
$s_j$. Let us consider for example that only $v_{i,j,k} \in$ AoI.
$s_i$ sends the packet to $s_j$ {\bf unless} both
\begin{itemize}
\item
$s_kD > s_iD$
\item
$\angle D s_j s_k < \angle D s_j s_i$, \xx or $\angle D s_j s_k = \angle D s_j s_i$ and ID of $s_k <$ ID of $s_i$ \yy
\end{itemize}
since it would mean that $s_k$ is better off than $s_i$ in sending the packet to $s_j$ according to MABRAVO$_D$ unicast routing algorithm.
\end{itemize}

Next section provides correctness proof for the MABRAVO$_R$ algorithm.

\begin{algorithm}[!tb]
\caption{Algorithm MABRAVO$_R$, executed by site $s_i$, having $D$ as source of the AoI-cast}
\label{alg:impl:inv}
\DontPrintSemicolon
\SetAlgoLined
Deliver the packet to $s_i$ \;
\ForEach{ $s_j$ neighbor of $s_i$}{
    \If{$Ds_i > Ds_j$}{
        Jump out to the main foreach cycle
    }
    Let $s_k$ and $s_l$ be the common neighbors of site $s_i$ and site $s_j$ \;
    \If{($v_{i,j,k} \notin AoI$) and ($v_{i,j,l} \notin AoI$)}{  
        \ForEach{$q \in \textbf{side}(AoI)$}{  
            \If{$q \cap v_{i,j,k} v_{i,j,l} \neq 0$}{ 
                Send packet to $s_j$\;
                Jump out to the main foreach cycle
            }
        }
    }
    \If{ $v_{i,j,k} \in AoI$ and $s_kD > s_iD$}{  
        \If {($\angle D s_j s_k < \angle D s_j s_i$) or ($\angle D s_j s_k = \angle D s_j s_i$ and ID of $s_k <$ ID of $s_i$)}{
            Jump out to the main foreach cycle
        }
    }
    \If{ $v_{i,j,l} \in AoI$ and $s_lD > s_iD$}{ 
        \If{ ($\angle D s_j s_l < \angle D s_j s_i$) or ($\angle D s_j s_l = \angle D s_j s_i$ and ID of $s_l <$ ID of $s_i$)}{
            Jump out to the main foreach cycle  
        }
    }
    Send packet to $s_j$\;
    Jump out to the main foreach cycle
}
\end{algorithm}

\subsection{Proof of correctness for the AoI-cast algorithm MABRAVO$_R$}\label{subsec:theomultproof}

\begin{theorem}
If $A$ computes a non-empty (local vision) $V_{C,A}$ of the Voronoi cell of $C$, then
$V_C \subseteq V_{C,A}$. Moreover, 
if $I_{C,A}$ exists, $I_C \subseteq I_{C,A}$.\label{th:rev:subset}\end{theorem}
\begin{proof}
Given two sites $C$ and $A$, the Voronoi area of $C$ in the local vision of $A$, called $V_{C,A}$, exists {\bf if and only if} $A$ and $C$ are neighbors.
Let us consider point $P$, and the algorithm that is used to decide if $P \in V_{C,A}$.
Let us call $S$ the set of all sites in the Voronoi tessellation (global vision of the overlay), and $S(A)$ the set of sites in the local vision
 of $A$, which contains $A$ and $A$'s neighbors. From the definition of a Voronoi cell, we have that:
\begin{displaymath}
P \in V_{C,A}\textrm{ {\bf if~and~only~if} }
 \forall s_i \in S(A): d(P,C) \le d(P,s_i)
\end{displaymath}
\begin{displaymath}
P \in V_C \textrm{ {\bf if~and~only~if} }
 \forall s_i \in S: d(P,C) \le d(P,s_i)
\end{displaymath}
Since the set of the sites in a local vision is a subset of the set of all the neighbors ($S(A) \subseteq S$), the set of conditions for $P \in V_{C,A}$ is subset of the set of conditions for $P \in V_C$, and $P \in V_C \Rightarrow P \in V_{C,A}$. Thus, $V_C \subseteq V_{C,A}$.
Considering now the intersection between the AoI and the Voronoi
 cells, since $I_C = V_C \bigcap \textrm{AoI}$ and
 $I_{C,A} = V_{C,A} \bigcap \textrm{AoI}$, and we just showed that $V_C \subseteq V_{C,A}$,
it holds that 
 $I_C \subseteq I_{C,A}$.
\end{proof}

\begin{theorem}
$Z_C(D)$ and $Z_{C,A}(D)$ are convex.
\end{theorem}
\begin{proof}
First of all, since both $Z_C(D)$ and $Z_{C,A}(D)$ are computed in the same way, 
the proof will be shown considering $Z_C(D)$ only, but it applies to both sets.
If $D \in I_C$, since $I_C$ is convex, all segments connecting points of $I_C$ to $D$ are internal to $I_C$, hence $Z_C(D) = I_C$, which is convex.
If $D \notin I_C$, building $Z_C(D)$ is analogous to applying a step of an Incremental Convex Hull algorithm (see for example Gift Wrapping\cite{jarvis73}, or Incremental Convex Hull\cite{Kallay84}), starting from $I_C$, which is convex and is the convex hull of its vertices, and adding the point $D$.
\end{proof}

\begin{theorem}Locus $S_C(D)$ and locus $S_{C,A}(D)$ are a connected component each.\label{th:rev:nohole}\end{theorem}
\begin{proof}
If $D \in I_C$, $S_C(D)$ is the empty set. Let us consider that $D \notin I_C$.
From the computation of $Z_C(D)$ using an Incremental Convex Hull algorithm \cite{jarvis73, Kallay84}, the $S_C(D)$ is constituted by the segments linking the vertices of $I_C$ that are not vertices of $Z_C(D)$, plus the two vertices of $I_C$ that were linked to $D$. Thus, $S_C(D)$ is a succession of adjacent segments, thus $S_C(D)$ is a connected component.
The proof regarding $S_{C,A}(D)$ is analogous.
\end{proof}

\begin{corollary}
Considering sites $A$, $B$ and $C$ that are mutual neighbors of each others, it holds
 that $B \in N_C(D) \bigwedge A \in N_C(D)$ {\bf if and only if}
 the common vertex of $V_A$, $V_B$ and $V_C$ lays into $S_C(D)$.
\end{corollary}

\begin{corollary}
Consider now local visions. The common vertex of $v_{A,B,C}$ three sites A, B and C that
 are mutual neighbors,
 is computed in the same way by the three sites. Thus,
 since the knowledge about the AoI is global, all $A$, $B$ and $C$ agree
 on the belonging of the common vertex to the segments of interest  $S_{C}(D)$, $S_{C,A}(D)$ and $S_{C,B}(D)$.
\end{corollary}

\begin{corollary}\label{co:agree}
As a consequence of the previous corollary, $A$, $B$ and $C$ agree on
\begin{itemize}
\item
$C$ sending a packet to $A$ -- or not -- to get to $D$ with MABRAVO$_D$ algorithm, and
\item
$A$ sending a packet to $C$ -- or not -- for a AoI-cast generated in $D$ with MABRAVO$_R$ algorithm.
\end{itemize}
\end{corollary}

\begin{theorem}Existence for MABRAVO$_R$ routes (each site in the AoI receives the packet at least once).\end{theorem}
\begin{proof}
Let us consider that $D$ generates a MABRAVO$_R$ AoI-cast,
that $A$ must receive the packet because its Voronoi cell owns points included in the AoI, and that $B$ decides not to forward a packet to $A$. Let us prove that there will be another site forwarding the packet to $A$.

Site $B$ can take the decision not to forward the packet
 to $A$ for two motivations:
\begin{itemize}
\item
Site $B \notin N_{A,B}(D)$, which is the set of the neighbors of $A$ that are towards $D$ in the local vision of $B$.
Since $I_{A,B} \supseteq I_A$ (see Theorem~\ref{th:rev:subset}), $B$ can not be in $N_A(D)$, and since the unicast route from $A$ to $D$ exists (see Theorem~\ref{th:uni:exi}), there must be at least another site
that will forward the packet to $A$ during MABRAVO$_R$ routing.
\item
Function $\angle D s_i B$ has not a minimum in $s_i = A$.
The cause is that one common neighbor of $A$ and $B$ (let us call it $C$) is into $N_{A,B}(D)$ and it has a smaller angle. Since $A$, $B$ and $C$ agree if $C \in N_{A,B}(D)$ (see Corollary~\ref{co:agree}), $C$ will either select itself to send the packet to $A$, or it will repeat the same reasoning for a common neighbor (let us call it $F$) of $C$ and $A$, but on the other side with respect to $B$. Since the number of neighbor of $A$ is finite, this chain will end up on a site (let call it $G$) that will actually send the packet to $A$.
\end{itemize}

Thus, if $B$ decides not to forward a packet to $A$, there will be at least another site that will forward the packet to $A$, and hence there exists at least one MABRAVO$_R$ route reaching $A$.
\end{proof}
\begin{figure}[!t]
\begin{center}
\includegraphics[width=0.90\linewidth]{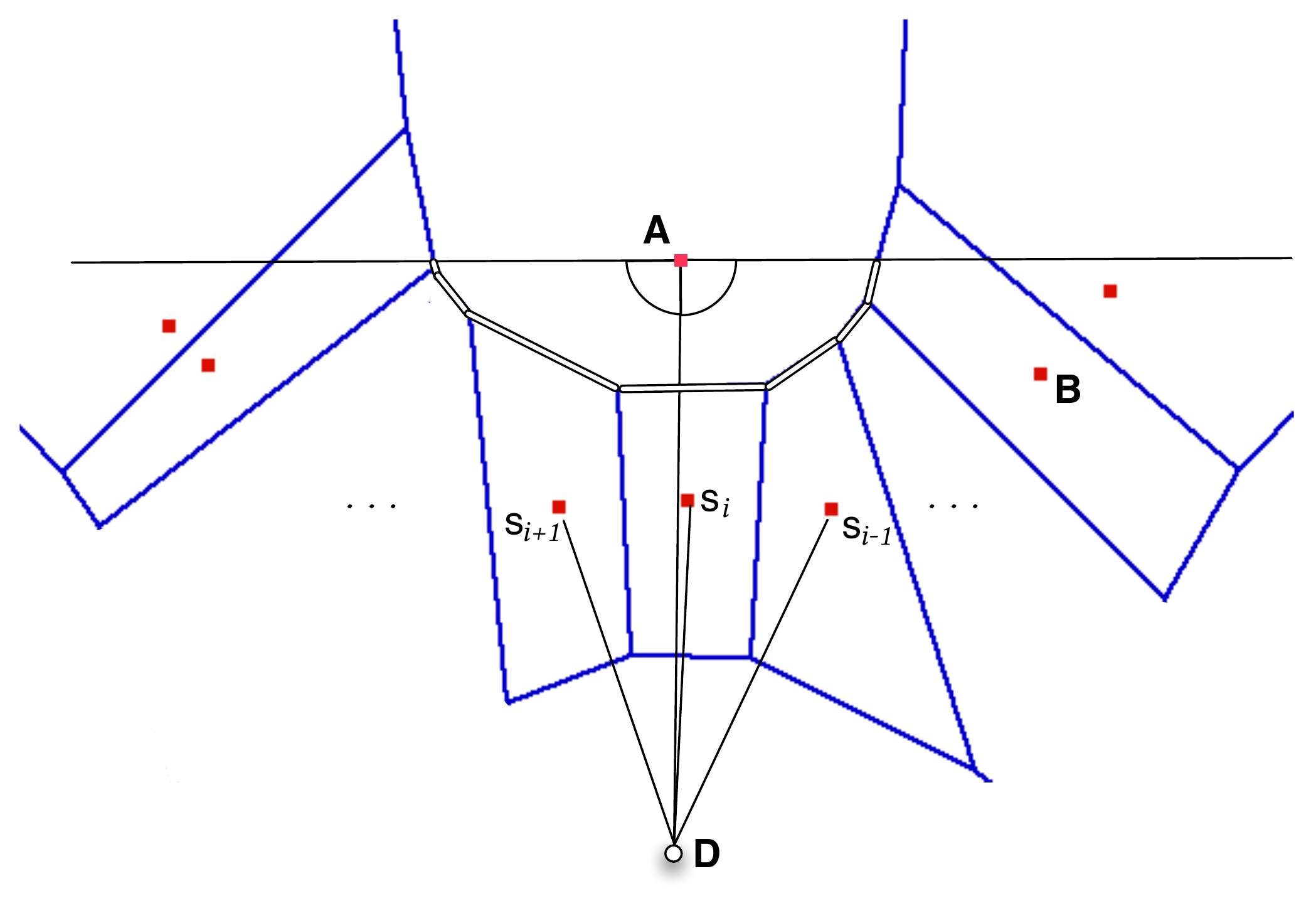}
\caption{Unicity of MABRAVO$_R$ AoI-cast routes.}\label{fig:Unicity}
\end{center}
\end{figure}
\begin{theorem}Unicity for MABRAVO$_R$ route (each site in the AoI receives the packet at most once).\label{th:rev:uni}\end{theorem}
\begin{proof}
Let us prove that a generic site $A$ can not receive a packet
 from more than one MABRAVO$_R$ route.
First of all, as a consequence of Theorem~\ref{theo:distances}, we have to consider that any site that will forward a message to $A$ should lay in the half of the plane that is closer to $D$ than $A$ itself.
Let us now suppose that site $B$ decides to
 send a packet to site $A$ during a MABRAVO$_R$ routing
 originated in point $D$. In the following, we use Fig.~\ref{fig:Unicity} as a possible representation of the situation.
A necessary condition is that $B \in N_{A,B}(D)$, and thus $B \in N_A(D)$.
Moreover, $B$ should see that $\angle D A B$ is smaller that the angle formed by any of its neighbors in $N_A(D)$. This fact implies that $B$ is the site in $N_A(D)$ that is closer to the half-plane $A-D$ bisector line, that connects $D$ and $A$. Otherwise, since Theorem \ref{th:rev:nohole} states that all the sites in $S_A(D)$ form a connected locus on the plane, there exists a series of other sites $s_1, ... , s_i$ that starts from a $B$'s neighbor $s_1$, and where a site $s_k$ is neighbor of $s_{k+1}$. These sites should be closer to the bisector line than $B$, as shown in Fig.~\ref{fig:Unicity}. However, this also implies that each of these sites forms an angle with $D$ and $A$ that is lower than $\angle D A B$. Therefore, $B$ cannot consider itself as the best candidate to forward a message to $A$, since it determines that at least one of its neighbors is a better candidate. The same decision would be taken by all the sites $s_1, ..., s_{i-1}$. This leads to the conclusion that only one site ($s_i$ in the example)  considers itself as the site in charge to deliver a message to $A$, thus demonstrating the unicity of the selection of a message forwarder in MABRAVO$_R$.
\end{proof}

\x

\subsection{Considerations about the complexity of the algorithms}\label{sec:complexity}

This section discusses briefly the complexity of the MABRAVO routing algorithms, both in terms of messages, and of computational complexity.

As proven in the previous subsections, the MABRAVO protocol suite allows for correct unicast and AoI-cast routing if the sites have up-to-date information regarding their own locations, and the location of their neighbors in the Voronoi diagram. As discussed in Section~\ref{subsec:contribution}, we consider that a topology maintenance algorithm is already in place in the network, since several solutions are already available in the literature~\cite{Alsalih08,pietrabissa2016distributed,beaumont2007voronet}. For example, the VoroNet~\cite{beaumont2007voronet} topology maintenance algorithm has a message complexity for each site that is proportional to the number of its neighbors. 
Since the expected number of neighbors of a site is lower than $6$~\cite{Aurenhammer}, the amortized message complexity for each site to maintain the topology is $O(1)$.

The message complexity for the MABRAVO$_R$ AoI-cast algorithm was proven to be optimal in Section~\ref{subsec:theomultproof}. The algorithm is able to create a routing tree over the sites in the AoI in a distributed manner, and the total number of messages is equal to the number of sites in the AoI minus $1$.

In the rest of this section, let us call $n$ the number of neighbors of a site $s_i$, and $m$ the number of sides defining the AoI. The computational complexity of both the MABRAVO$_D$ algorithm (Algorithm~\ref{alg:impl:direct}) and the MABRAVO$_R$ algorithm (Algorithm~\ref{alg:impl:inv}) depends (i) on the expected number of neighbors of each site $s_i$ being less than $6$~\cite{Aurenhammer}, thus $O(1)$; (ii) on the operations discussed in Section~\ref{subsec:req} being able to access the set of neighbors of each site $s_i$ in linear time in the number of neighbors; (iii) on the fact that the number of sides of the AoI is an external parameter set by the user defining the AoI, 
and in most applications this value can be considered sufficiently small. For instance, in many location-based applications, areas/regions of interest are defined as rectangles (e.g.~\cite{ko2000location,6714420,seada2006efficient}).

With regards to the MABRAVO$_D$ algorithm (Algorithm~\ref{alg:impl:direct}), its first loop of the algorithm is repeated for each neighbor of a site $s_i$ (thus $O(n)$ times), and each time it accesses the list of neighbors of $s_i$ (complexity $O(n)$) and it compares the location of each neighbor with each AoI side, whose cardinality is $O(m)$. Thus, the complexity of the first loop is $O(n^2 m)$. The second loop is executed when no vertices of the Voronoi cell of $s_i$ are into the AoI, it is repeated for each neighbor of $s_i$ ($O(n)$ times), its internal loop is repeated for each side of the AoI ($O(m)$ times), thus the complexity of the second loop is $O(n m)$. The third loop is repeated over the neighbors of $s_i$, which are $O(n)$, it performs only operations with constant complexity, thus the complexity of the third loop is $O(n)$. Thus, the computational complexity of the MABRAVO$_D$ algorithm is $O(n^2 m + n m + n) = O(n^2 m)$. 
It is worth noticing that, in the average case, $n$ has an expected value that is equal or less than the constant $6$. Therefore, in the average case the complexity reduces to $O(m)$. In addition, as observed in the previous paragraph, $m$ has generally a small value, thus leading to an overall low complexity.

With regards of the MABRAVO$_R$ algorithm (Algorithm~\ref{alg:impl:inv}), its most external loop is repeated for each neighbor of $s_i$, thus $O(n)$ times. The first condition in the algorithm (the neighbors being located outside the AoI) requires to repeat basic geometric operations for each side of the AoI (thus, $O(m)$ times), then it extracts $v_{i,j,k}$ and $v_{i,j,l}$ (cost $O(n)$), and then executes a loop on the sides of the AoI (thus, $O(m)$ times), each time performing operations having constant execution time. If at least one of the neighbors of $s_i$ is located into the AoI, whose test costs $O(m)$, the algorithm executes in the worst case the two \texttt{if} clauses, which require to perform basic geometric operations. Thus, the computational complexity of the MABRAVO$_R$ algorithm is $O(n (m (n+m)+m) ) = O(n^2m + n m^2)$.
Applying the same reasoning used for MABRAVO$_D$, in the average case the value of $n$ is equal or less than $6$. The overall complexity is thus a function of $m$ only (i.e.: $O(m^2)$), where $m$ is generally a small value.

\y

\section{Evaluation of the MABRAVO suite}\label{sec:sims}

This section describes the experimental evaluation of the algorithms of the MABRAVO suite. The evaluation is made through a simulation implementation of the proposed solution. The results shown in the rest of this section have been selected in order to better highlight the features of MABRAVO and to allow to experimentally corroborate the correctness of the algorithms. 

\subsection{Implementation}

In this section, we provide a description of the simulator we used to derive the results presented in the rest of the section. This description makes it possible to use the related software. Thus, it allows to make the results we present verifiable and fully reproducible by the scientific community. The simulator is available on github (https://github.com/michelealbano/mabravo) and it was \x published \y on Code Ocean (DOI: 10.24433/CO.1722184.v1).

The algorithms were implemented using the Java programming language and are accessible as a supplementary material of this paper. This subsection describes the code, shows how to compile it, and what it does when executed.

The implementation makes use of the VAST library, 
a well-known library used in the literature to help evaluate Voronoi-based solutions (e.g., the proof-of-concept in~ \cite{VON}).
The novel code comprises 4 classes:
\begin{itemize}
    \item {\bf AreaOfInterest} maintains a convex AoI on the plane. When instantiated, it receives a number of points in the plane, and it makes use of the Gift Wrapping\cite{jarvis73} algorithm to organize them as a clockwise sequence of points that define the AoI;
    \item {\bf VoronoiArea} is a thin wrapper over the mechanisms provided by the VON codebase;
    \item {\bf VoronoiNetwork} implements all the routing algorithms of the MABRAVO suite, and a breadth first visit that is used to compute e.g. the number of sites that lie into an AoI;
    \item {\bf Mabravo} contains all the parameters that are used to specify the simulations to be performed, it drives the execution of the experiments, it can provide a simple graphical representation of the routing processes, and it can compute performance parameters to summarize the results of routing processes.
\end{itemize}
To deploy the system, it is sufficient to issue a \texttt{make} command on the command line. After that, the software can be executed in two modes.

\begin{figure*}[!t]
\centering
\centering
\includegraphics[width=2.75in]{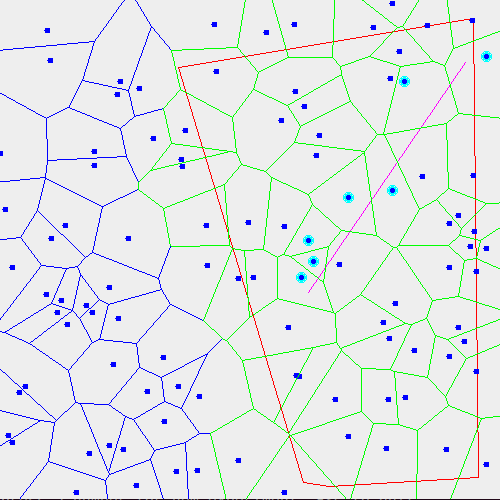}
\hfil
\centering
\includegraphics[width=2.75in]{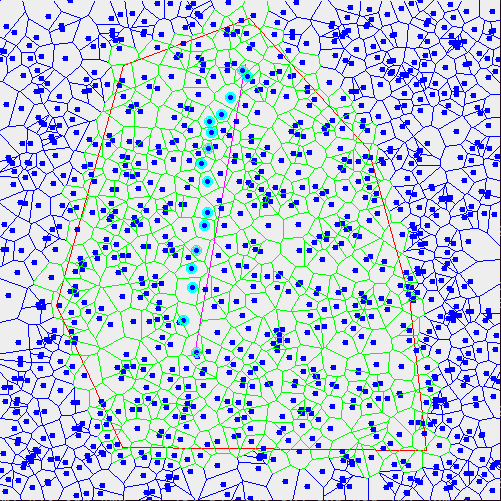}
\caption{Routing process on a network comprising 100 (left) and 1000 (right) sites.}
\label{fig:simroutes}
\end{figure*}

The first one is named the {\bf graphical} mode. It is executed if the user provides $3$ parameters on the command line. An example of the invocation of this mode of execution is \texttt{java -cp mabravo-1.1.0.jar mabravo.Mabravo 100 10 1000}. 
These parameters are: the number of sites; the number of points defining the AoI; a random seed.
This execution mode allows the user to have a visual representation of the system and of the execution of the MABRAVO protocol suite.
Specifically, the Mabravo application creates a number of sites coordinates at random in the plane, and it instantiates an AoI with a given number of vertices. After that, it selects two points at random in the AoI and performs a routing process from the first to the second. Finally, a graphical representation of the process is provided to the user. If the user presses the \texttt{return} key, the process will start once again with new random coordinates.

The second execution mode  is the {\bf batch} mode. It does not provide any graphical representation, but it is designed to allow to perform a series of different simulations of the system, and to extract performance indicators. To perform the simulations, five parameters are passed from the command line when invoking the main class
(e.g.: \texttt{java -cp mabravo-1.1.0.jar mabravo.Mabravo 100 10 100 10 100}).
These parameters are: the number of sites; the number of points defining the AoI; the number of routing processes to be performed over each network; the number of networks to be simulated; a random seed.

For each of the routing processes, the system print out data to both evaluate the proposed solution, and to compare it against an ``oracle", i.e. a solution that computes the routing tree by exploiting a breadth first visit and the full knowledge about the structure and topology of the system. The values that are used for evaluation and comparison that are returned by the simulator are:
\begin{itemize}
    \item site where the unicast routing process starts;
    \item site where the unicast routing end / site where the AoI-cast starts;
    \item number of nodes in the whole network;
    \item nodes in the AoI;
    \item number of hops for the unicast routing using the ``oracle";
    \item average length of the AoI-cast routes using the ``oracle";
    \item average length of the AoI-cast using MABRAVO$_R$;
    \item unicast route computed using MABRAVO$_D$.
\end{itemize}
These are the values used in the next section for the overall evaluation of MABRAVO.

When executing the MABRAVO$_D$ algorithm, the simulator verifies that the routes goes from the source to the destination, and that no site outside the AoI is reached by the routing process. When executing the MABRAVO$_R$ algorithm the simulator verifies that all the sites in the AoI receive the message once, and that no site outside the Aoi receives the message.

\subsection{Results}

In the following, we present the results obtained by using the simulator described in the previous section. The experiments focus on the more relevant characteristics of the MABRAVO suite, and aimed to:

\begin{itemize}
    \item verify that the MABRAVO$_R$ algorithm is always able to deliver a packet from a source site in the AoI to a destination site in the AoI, without using relays outside the AoI;
    \item proof that the MABRAVO$_R$ algorithm \x sends \y a packet to all the sites in the AoI and no one else, and that the sites receive the packet only once;
    \item compare the length of both MABRAVO$_D$ and MABRAVO$_R$ routes against the ``oracle", as defined in the previous subsection.
\end{itemize}

In order to present an example of the execution of MABRAVO,
Figure~\ref{fig:simroutes} shows the output of the simulator when executed in {\em graphical} mode, for networks comprising $100$ and $1000$ sites, respectively. In both these cases,
the AoI is defined by $10$ points. The meanings of the colors in the figure are the following: the red lines are the borders of the AoI; the Voronoi cells of sites in the AoI have green borders, while sites outside the AoI have Voronoi cells with blue borders; the magenta line connects source and destination of the unicast routing process (the $100$ sites case shows clearly that source and destination points do not have to be co-located with a site), and the sites touched by the routing process are highlighted with cyan circles.

In order to present a meaningful evaluation, the results we show are the average of a series of repeated executions of the system. Specifically,
the system was run in {\em batch} mode to simulate $100$ different networks and perform $100$ MABRAVO$_D$ and MABRAVO$_R$ routing processes on them.
The experiments were executed on networks of $100$ and $1000$ sites, respectively. \x In both the scenarios, the AoIs are defined by $10$ points. \y

\begin{figure*}[!t]
\centering
\centering
\includegraphics[width=2.75in]{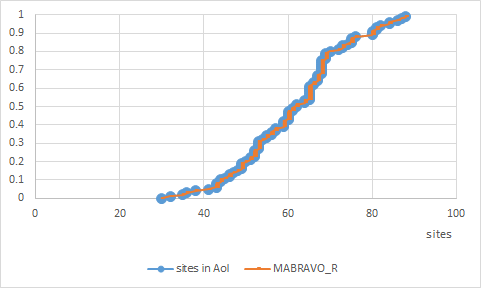}
\hfil
\centering
\includegraphics[width=2.75in]{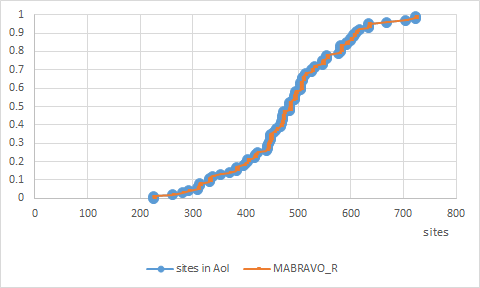}
\caption{CDF of the number of sites in the AoI, and of the number of sites receiving an AoI-cast with MABRAVO$_R$, in networks comprising 100 (left) and 1000 (right) sites.}
\label{fig:simaoisites}
\end{figure*}

\x
For the sake of clarity, in the following we will show experimental results by means of their Cumulative Distribution Function (CDF)~\cite{montgomery2010applied}, meaning that the graphs will show the possible values of the variable under study on the $x$ axis, and the probability that the output of an experiment is less or equal to the value on the $y$ axis.
\y

The simulator confirmed that, in all the scenarios, the MABRAVO suite sent messages exclusively to sites that are included in the AoI.
\x With regard to the AoI-cast algorithm, we counted the number of sites in the AoI from a global vision, and we compared that to the number of sites that receive the AoI-cast message with MABRAVO$_R$ algorithm. We present the results of the experiments in Figure~\ref{fig:simaoisites}, whose overlapping curves ensure that MABRAVO$_R$ delivers the AoI-cast message exactly  to the sites comprised (fully or partially) within the AoI. \y

\begin{figure*}[!t]
\centering
\centering
\includegraphics[width=2.75in]{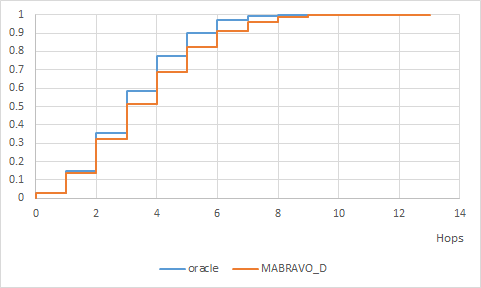}
\hfil
\centering
\includegraphics[width=2.75in]{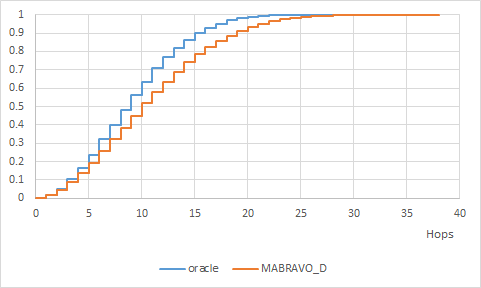}
\caption{CDF of the number of hops in the unicast routing in networks comprising 100 (left) and 1000 (right) sites.}
\label{fig:simucast}
\end{figure*}

\begin{figure*}[!t]
\centering
\centering
\includegraphics[width=2.75in]{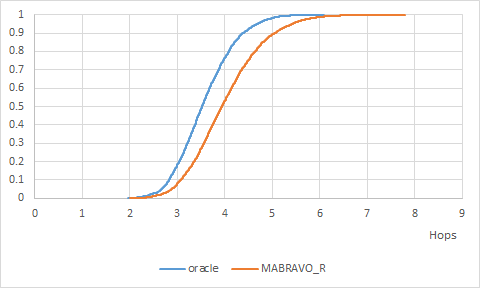}
\hfil
\centering
\includegraphics[width=2.75in]{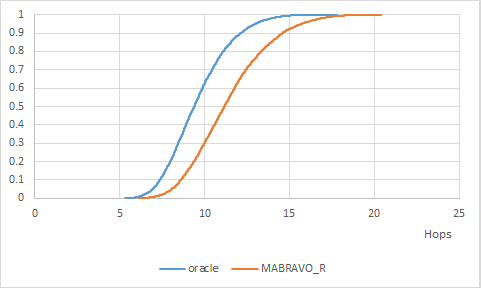}
\caption{CDF of the average number of hops in the AoI-cast in networks comprising 100 (left) and 1000 (right) sites.}
\label{fig:simmcast}
\end{figure*}

We first show the results related to the unicast protocol of MABRAVO, i.e. MABRAVO$_D$. As we anticipated,
when performing unicast communication, using the MABRAVO$_D$ algorithm 
the routing process was always able to route the packets using only sites in the AoI as relays. Figure~\ref{fig:simucast} shows 
the CDF of the route length when MABRAVO$_D$ is employed, presenting the results for
networks of $100$ and $1000$ sites. The results are compared with the ``oracle", showing two remarkably close and similar behaviors. It is worth noticing that the ``oracle" can exploit full knowledge of the geometry of the system, while MABRAVO$_D$ can only rely on partial and local information, and it is the result of autonomous decisions of independent components of the system.

The same considerations are valid for the AoI-cast protocol MABRAVO$_R$. Also in this case, the AoI-cast performed by the MABRAVO$_R$ algorithm is able to reach each site in the AoI with a packet by crossing only sites comprised in the AoI.
Figure~\ref{fig:simmcast} compares the length of the routes for AoI-cast routing by showing the CDF of the average length of the routes, when the two algorithms are used on networks of $100$ and $1000$ sites.

\x
\section{Related Works}\label{sec:related}
Geometric routing techniques are among the most commonly used strategy for forwarding messages and information in Voronoi networks. Compass Routing is the most important solution in this class of routing protocols. 
The use of Compass Routing in Voronoi networks has been first proposed in \cite{Kranakis}, which considers a connected graph and 
assumes that a message is generated at one of its nodes $n$ with the goal to reach a destination node $d$. \cite{Kranakis}
shows that the best strategy is to look at the edges incident in $n$ and
choose the edge whose slope is minimal with respect to the segment
connecting $n$ and the destination $d$. \cite{Kranakis} also shows that 
while Compass Routing is not cycle free for general graphs, it can always find a finite path between two
nodes of a Delaunay Triangulation. The work in~\cite{liebeherr2002application} suggests to exploit Compass 
Routing to define a Spanning Tree supporting an application level multicast.

This class of solutions is very relevant in this field, since it forms the basis of many other routing protocols on Voronoi networks. However, they do not face the problems related to multicast routing within delimited AoIs. These issues are faced by the approaches described in the following of this section.

As far as its applications are concerned, Voronoi networks have recently been exploited in several contexts, but  mostly for the definition of routing algorithms in sensor and wireless sensor networks.

An approach which shares some features with our proposal is introduced in \cite{grey2018towards}. Indeed,  Overlay Geocast, besides forwarding messages toward a given destination, is also able to route messages to all the nodes belonging to a given area $A$. To reach a node in $A$, it exploits greedy routing, then each node in $A$ forwards the message to all its neighbors in $A$ and discards
duplicates by utilizing a Bloom Filters. This implies a large number of unnecessary messages, while MABRAVO totally avoids unnecessary communications.

\cite{xing2004greedy} introduces {\em sensing-covered networks}, which are networks where every
point in a geographic area must be within the
sensing range of at least one sensor. The paper introduces a new routing algorithm, {\em Bounded Voronoi Greedy
Forwarding (BVGF)}, that combines Greedy Forwarding and  and
Voronoi diagrams. When a node forwards a packet,  it considers its eligible neighbors, where a neighbor is eligible if the line segment joining
the source and the destination intersects the Voronoi region of the neighbor or coincides
with one of the boundaries of the Voronoi region. BVGF chooses the neighbor that has the shortest Euclidean distance to the destination among all eligible neighbors.

\cite{dai2016novel} proposes a Voronoi diagram based on semi-distributed algorithms for coverage holes detection in WSNs. The Voronoi diagram is built by considering the location of the sensor nodes which have the task of monitoring and collecting information on the Region of Interest (ROI). Furthermore, the proposed algorithms decide if there are holes in the ROI.

\cite{wan2019multi} investigates the use of wireless sensor networks in IoT environments, to
monitor and collect data in some geographic
area. In this case, spatial range queries with location constraints are employed. To reduce the communication cost and the storage requirements, the work presents an energy- and time-efficient multidimensional
data indexing scheme which exploits a Voronoi tessellation.

\y

\section{Conclusions \& Future Works}\label{sec:conc}
This paper presents an algorithm to perform AoI-cast in
 Voronoi-based distributed networks. The proposed solution is
 able to construct AoI-cast trees in a completely
 distributed manner, where each agent that supervises a Voronoi cell knows only its own coordinates and the ones of its immediate neighbors. Working in totally decentralized manner,  the proposed algorithms are able to deliver packets by reaching all (and only) the sites in a convex AoI, thus requiring a minimal number of messages.
In this work we gave a formal specification of MABRAVO, as well as a formal demonstration of its properties.

\x An open issue to investigate is whether it exists an algorithm that is able to minimize the route lengths depth, while preserving the optimal properties of the algorithms presented in this work (only local knowledge required, minimum number of packets). \y
More future work has been planned:

\x
\begin{itemize}
\item Port the algorithm to another language (currently it is pure Java), which comprises studying how to use the existing libraries of target language. For example, this will allow to access high-performance libraries for the computation of the Voronoi diagrams,
and possibly to use hardware accelerations;
\item Implement the algorithm into a mainstream system simulator, such as ns-3, to experiment the algorithm against physical properties of the wireless links, loss of packets, presence of metal walls and other obstacles that hinder communication;
\item Study the effect of mobility and high churn of units, and the behavior of the algorithms against obsolete information regarding a unit's neighbors;
\item Implement a testbed, where the algorithms are used to enable communication between robots in industrial settings, for example basing the exchange of messages between neighbors over the Arrowhead framework~\cite{delsing2017arrowhead}.
\end{itemize}
\y

\bibliographystyle{ieeetr}

\begin{thebibliography}{10}

\bibitem{iopComcom}
M.~Conti and A.~Passarella, ``The internet of people: A human and data-centric
  paradigm for the next generation internet,'' {\em Computer Communications},
  vol.~131, pp.~51--65, 2018.

\bibitem{sayed18}
H.~El-Sayed, S.~Sankar, M.~Prasad, D.~Puthal, A.~Gupta, M.~Mohanty, and C.-T.
  Lin, ``Edge of things: The big picture on the integration of edge, iot and
  the cloud in a distributed computing environment,'' {\em IEEE Access},
  vol.~6, pp.~1706--1717, 2017.

\bibitem{sot16}
{\"O}.~U. Akg{\"u}l and B.~Canberk, ``Self-organized things (sot): An energy
  efficient next generation network management,'' {\em Computer
  Communications}, vol.~74, pp.~52--62, 2016.

\bibitem{wang17}
S.~Wang, X.~Zhang, Y.~Zhang, L.~Wang, J.~Yang, and W.~Wang, ``A survey on
  mobile edge networks: Convergence of computing, caching and communications,''
  {\em IEEE Access}, vol.~5, pp.~6757--6779, 2017.

\bibitem{iop2020human}
M.~Mordacchini, M.~Conti, A.~Passarella, and R.~Bruno, ``Human-centric data
  dissemination in the iop: Large-scale modeling and evaluation,'' {\em ACM
  Transactions on Autonomous and Adaptive Systems (TAAS)}, vol.~14, no.~3,
  pp.~1--25, 2020.

\bibitem{delsing2017arrowhead}
J.~Delsing, P.~Varga, L.~Ferreira, M.~Albano, P.~P. Pereira, J.~Eliasson, and
  H.~Derhamy, ``The arrowhead framework architecture,'' {\em IoT Automation:
  Arrowhead Framework}, pp.~45--91, 2017.

\bibitem{bellavista18}
P.~Bellavista, S.~Chessa, L.~Foschini, L.~Gioia, and M.~Girolami,
  ``Human-enabled edge computing: Exploiting the crowd as a dynamic extension
  of mobile edge computing,'' {\em IEEE Communications Magazine}, vol.~56,
  no.~1, pp.~145--155, 2018.

\bibitem{mordacchini2015crowdsourcing}
M.~Mordacchini, A.~Passarella, M.~Conti, S.~M. Allen, M.~J. Chorley, G.~B.
  Colombo, V.~Tanasescu, and R.~M. Whitaker, ``Crowdsourcing through cognitive
  opportunistic networks,'' {\em ACM Transactions on Autonomous and Adaptive
  Systems (TAAS)}, vol.~10, no.~2, pp.~1--29, 2015.

\bibitem{8890658}
Q.~{Zheng}, J.~{Jin}, T.~{Zhang}, J.~{Li}, L.~{Gao}, and Y.~{Xiang},
  ``Energy-sustainable fog system for mobile web services in
  infrastructure-less environments,'' {\em IEEE Access}, vol.~7,
  pp.~161318--161328, 2019.

\bibitem{STRAYER20181}
T.~Strayer, S.~Nelson, A.~Caro, J.~Khoury, B.~Tedesco, O.~DeRosa, C.~Clark,
  K.~Sadeghi, M.~Matthews, J.~Kurzer, P.~Lundrigan, V.~Kawadia, D.~Ryder,
  K.~Gremban, and W.~Phoel, ``Content sharing with mobility in an
  infrastructure-less environment,'' {\em Computer Networks}, vol.~144, pp.~1
  -- 16, 2018.

\bibitem{Alsharoa18}
A.~Alsharoa and M.~Yuksel, ``Uav-direct: Facilitating d2d communications for
  dynamic and infrastructure-less networking,'' in {\em Proceedings of the 4th
  ACM Workshop on Micro Aerial Vehicle Networks, Systems, and Applications},
  DroNet'18, p.~57–62, Association for Computing Machinery, 2018.

\bibitem{8379320}
L.~{Chancay-Garc\'{i}a}, E.~{Hern\'{a}ndez-Orallo}, P.~{Manzoni}, C.~T.
  {Calafate}, and J.~{Cano}, ``Evaluating and enhancing information
  dissemination in urban areas of interest using opportunistic networks,'' {\em
  IEEE Access}, vol.~6, pp.~32514--32531, 2018.

\bibitem{Haosheng18}
H.~Huang, G.~Gartner, J.~M. Krisp, M.~Raubal, and N.~V. de~Weghe, ``Location
  based services: ongoing evolution and research agenda,'' {\em Journal of
  Location Based Services}, vol.~12, no.~2, pp.~63--93, 2018.

\bibitem{Liu2014}
E.~S. Liu and G.~K. Theodoropoulos, ``Interest management for distributed
  virtual environments: A survey,'' {\em ACM computing surveys (CSUR)},
  vol.~46, no.~4, p.~51, 2014.

\bibitem{vrsense}
B.~Jung, S.~Lim, J.~Chae, and C.~Pu, ``Vrsense: Validity region sensitive query
  processing strategies for static and mobile point-of-interests in manets,''
  {\em Computer Communications}, vol.~116, pp.~132--146, 2018.

\bibitem{mqry}
B.~Jung, S.~Lim, J.~Chae, C.~Pu, and M.~Min, ``Mqry: Elastic validity region
  for querying mobile point-of-interests in infrastructure-less networks,''
  {\em IEEE Systems Journal}, 2019.

\bibitem{Qi18}
J.~Qi, R.~Zhang, C.~S. Jensen, K.~Ramamohanarao, and J.~HE, ``Continuous
  spatial query processing: A survey of safe region based techniques,'' {\em
  ACM Comput. Surv.}, vol.~51, May 2018.

\bibitem{ghaffari2010necessity}
M.~Ghaffari, B.~Hariri, and S.~Shirmohammadi, ``On the necessity of using
  delaunay triangulation substrate in greedy routing based networks,'' {\em
  IEEE Communications Letters}, vol.~14, no.~3, pp.~266--268, 2010.

\bibitem{wan19}
S.~Wan, Y.~Zhao, T.~Wang, Z.~Gu, Q.~H. Abbasi, and K.-K.~R. Choo,
  ``Multi-dimensional data indexing and range query processing via voronoi
  diagram for internet of things,'' {\em Future Generation Computer Systems},
  vol.~91, pp.~382--391, 2019.

\bibitem{pietra19}
A.~Pietrabissa and F.~Liberati, ``Dynamic distributed clustering in wireless
  sensor networks via voronoi tessellation control,'' {\em International
  Journal of Control}, vol.~92, no.~5, pp.~1001--1014, 2019.

\bibitem{CRPD18}
S.~Wang, J.~Yu, M.~Atiquzzaman, H.~Chen, and L.~Ni, ``Crpd: a novel clustering
  routing protocol for dynamic wireless sensor networks,'' {\em Personal and
  Ubiquitous Computing}, vol.~22, no.~3, pp.~545--559, 2018.

\bibitem{underwater17}
E.~P. C{\^a}mara~J{\'u}nior, L.~F. Vieira, and M.~A. Vieira, ``Scheduling nodes
  in underwater networks using voronoi diagram,'' in {\em Proceedings of the
  20th ACM International Conference on Modelling, Analysis and Simulation of
  Wireless and Mobile Systems}, pp.~245--252, ACM, 2017.

\bibitem{underwater18}
S.~Wang, T.~L. Nguyen, and Y.~Shin, ``Data collection strategy for magnetic
  induction based monitoring in underwater sensor networks,'' {\em IEEE
  Access}, vol.~6, pp.~43644--43653, 2018.

\bibitem{taxi16}
D.~Zhang, J.~Wan, Z.~He, S.~Zhao, K.~Fan, S.~O. Park, and Z.~Jiang,
  ``Identifying region-wide functions using urban taxicab trajectories,'' {\em
  ACM Transactions on Embedded Computing Systems (TECS)}, vol.~15, no.~2,
  p.~36, 2016.

\bibitem{v2v18}
G.~Li, B.~He, and A.~Du, ``A traffic congestion aware vehicle-to-vehicle
  communication framework based on voronoi diagram and information
  granularity,'' {\em Peer-to-Peer Networking and Applications}, vol.~11,
  no.~1, pp.~124--138, 2018.

\bibitem{ghaffari2014dynamic}
M.~Ghaffari, B.~Hariri, S.~Shirmohammadi, and D.~T. Ahmed, ``A dynamic
  networking substrate for distributed mmogs,'' {\em IEEE Transactions on
  Emerging Topics in Computing}, vol.~3, no.~2, pp.~289--302, 2014.

\bibitem{mordacchini2010hivory}
M.~Mordacchini, L.~Ricci, L.~Ferrucci, M.~Albano, and R.~Baraglia, ``Hivory:
  Range queries on hierarchical voronoi overlays,'' in {\em 2010 IEEE Tenth
  International Conference on Peer-to-Peer Computing (P2P)}, pp.~1--10, IEEE,
  2010.

\bibitem{abdellatif2019edge}
A.~A. Abdellatif, A.~Mohamed, C.~F. Chiasserini, M.~Tlili, and A.~Erbad, ``Edge
  computing for smart health: Context-aware approaches, opportunities, and
  challenges,'' {\em IEEE Network}, vol.~33, no.~3, pp.~196--203, 2019.

\bibitem{an2018context}
C.~An, C.~Wu, T.~Yoshinaga, X.~Chen, and Y.~Ji, ``A context-aware edge-based
  vanet communication scheme for its,'' {\em sensors}, vol.~18, no.~7, p.~2022,
  2018.

\bibitem{liao2019learning}
H.~Liao, Z.~Zhou, X.~Zhao, L.~Zhang, S.~Mumtaz, A.~Jolfaei, S.~H. Ahmed, and
  A.~K. Bashir, ``Learning-based context-aware resource allocation for
  edge-computing-empowered industrial iot,'' {\em IEEE Internet of Things
  Journal}, vol.~7, no.~5, pp.~4260--4277, 2019.

\bibitem{VON}
S.-Y. Hu, J.-F. Chen, and T.-H. Chen, ``Von: a scalable peer-to-peer network
  for virtual environments,'' {\em IEEE Network}, vol.~20, no.~4, pp.~22--31,
  2006.

\bibitem{Voraque}
M.~Albano, L.~Ricci, M.~Baldanzi, and R.~Baraglia, ``Voraque: Range queries on
  voronoi overlays,'' in {\em 2008 IEEE Symposium on Computers and
  Communications}, pp.~495--500, IEEE, 2008.

\bibitem{Buyukkaya2008}
E.~Buyukkaya and M.~Abdallah, ``Data management in voronoi-based p2p gaming,''
  in {\em 2008 5th IEEE Consumer Communications and Networking Conference},
  pp.~1050--1053, IEEE, 2008.

\bibitem{Alsalih08}
W.~Alsalih, K.~Islam, Y.~N\'{u}\~{n}ez Rodr\'{\i}guez, and H.~Xiao,
  ``Distributed voronoi diagram computation in wireless sensor networks,'' in
  {\em Proceedings of the Twentieth Annual Symposium on Parallelism in
  Algorithms and Architectures}, SPAA '08, p.~364, Association for Computing
  Machinery, 2008.

\bibitem{pietrabissa2016distributed}
A.~Pietrabissa, F.~Liberati, and G.~Oddi, ``A distributed algorithm for ad-hoc
  network partitioning based on voronoi tessellation,'' {\em Ad Hoc Networks},
  vol.~46, pp.~37--47, 2016.

\bibitem{beaumont2007voronet}
O.~Beaumont, A.-M. Kermarrec, L.~Marchal, and {\'E}.~Rivi{\`e}re, ``Voronet: A
  scalable object network based on voronoi tessellations,'' in {\em 2007 IEEE
  International Parallel and Distributed Processing Symposium}, pp.~1--10,
  IEEE, 2007.

\bibitem{Beaumont07}
O.~Beaumont, A.-M. Kermarrec, L.~Marchal, and {\'E}.~Rivi{\`e}re, ``Voronet: A
  scalable object network based on voronoi tessellations,'' in {\em 2007 IEEE
  International Parallel and Distributed Processing Symposium}, pp.~1--10,
  IEEE, 2007.

\bibitem{albanoCSE}
M.~Albano, R.~Baraglia, M.~Mordacchini, and L.~Ricci, ``Efficient broadcast on
  area of interest in voronoi overlays,'' in {\em 2009 International Conference
  on Computational Science and Engineering}, vol.~1, pp.~224--231, IEEE, 2009.

\bibitem{Aurenhammer}
F.~Aurenhammer, ``Voronoi diagrams -- a survey of a fundamental geometric data
  structure,'' {\em ACM Computing Surveys (CSUR)}, vol.~23, no.~3,
  pp.~345--405, 1991.

\bibitem{golodoniuc2017distributed}
P.~Golodoniuc, N.~J. Car, and J.~Klump, ``Distributed persistent identifiers
  system design,'' {\em Data Science Journal}, vol.~16, 2017.

\bibitem{mahalle2020rethinking}
P.~N. Mahalle, G.~Shinde, and P.~M. Shafi, ``Rethinking decentralised
  identifiers and verifiable credentials for the internet of things,'' in {\em
  Internet of Things, Smart Computing and Technology: A Roadmap Ahead},
  pp.~361--374, Springer, 2020.

\bibitem{kortesniemi2019improving}
Y.~Kortesniemi, D.~Lagutin, T.~Elo, and N.~Fotiou, ``Improving the privacy of
  iot with decentralised identifiers (dids),'' {\em Journal of Computer
  Networks and Communications}, vol.~2019, 2019.

\bibitem{perkins2003ad}
C.~Perkins, E.~M. Royer, and S.~Das, ``Ad-hoc on demand distance vector routing
  (aodv),'' tech. rep., Internet-Draft, November 1997.
  draft-ietf-manet-aodv-00. txt, 2003.

\bibitem{jarvis73}
R.~A. Jarvis, ``On the identification of the convex hull of a finite set of
  points in the plane,'' {\em Information processing letters}, vol.~2, no.~1,
  pp.~18--21, 1973.

\bibitem{Kallay84}
M.~Kallay, ``The complexity of incremental convex hull algorithms in rd,'' {\em
  Information Processing Letters}, vol.~19, no.~4, p.~197, 1984.

\bibitem{ko2000location}
Y.-B. Ko and N.~H. Vaidya, ``Location-aided routing (lar) in mobile ad hoc
  networks,'' {\em Wireless networks}, vol.~6, no.~4, pp.~307--321, 2000.

\bibitem{6714420}
R.~{Jiang}, Y.~{Zhu}, T.~{He}, Y.~{Liu}, and L.~M. {Ni}, ``Exploiting
  trajectory-based coverage for geocast in vehicular networks,'' {\em IEEE
  Transactions on Parallel and Distributed Systems}, vol.~25, no.~12,
  pp.~3177--3189, 2014.

\bibitem{seada2006efficient}
K.~Seada and A.~Helmy, ``Efficient and robust geocasting protocols for sensor
  networks,'' {\em Computer Communications}, vol.~29, no.~2, pp.~151--161,
  2006.

\bibitem{montgomery2010applied}
D.~C. Montgomery and G.~C. Runger, {\em Applied statistics and probability for
  engineers}.
\newblock John Wiley \& Sons, 2010.

\bibitem{Kranakis}
E.Kranakis, H.Singh, and J.Urrutia, ``Compass routing on geometric networks,''
  in {\em In Proceedings of 11th Can. Conf. on Computational Geometry, CCCG},
  August 1999.

\bibitem{liebeherr2002application}
J.~Liebeherr, M.~Nahas, and W.~Si, ``Application-layer multicasting with
  delaunay triangulation overlays,'' {\em IEEE Journal on Selected Areas in
  Communications}, vol.~20, no.~8, pp.~1472--1488, 2002.

\bibitem{grey2018towards}
M.~Grey, M.~Theil, M.~Rossberg, and G.~Schaefer, ``Towards voronoi-based backup
  routing for large-scale distributed applications,'' in {\em 2018 Global
  Information Infrastructure and Networking Symposium (GIIS)}, pp.~1--5, IEEE,
  2018.

\bibitem{xing2004greedy}
G.~Xing, C.~Lu, R.~Pless, and Q.~Huang, ``On greedy geographic routing
  algorithms in sensing-covered networks,'' in {\em Proceedings of the 5th ACM
  international symposium on Mobile ad hoc networking and computing},
  pp.~31--42, 2004.

\bibitem{dai2016novel}
G.~Dai, H.~Lv, L.~Chen, B.~Zhou, and P.~Xu, ``A novel coverage holes discovery
  algorithm based on voronoi diagram in wireless sensor networks,'' {\em Int.
  J. Hybrid Inf. Technol}, vol.~9, no.~3, pp.~273--282, 2016.

\bibitem{wan2019multi}
S.~Wan, Y.~Zhao, T.~Wang, Z.~Gu, Q.~H. Abbasi, and K.-K.~R. Choo,
  ``Multi-dimensional data indexing and range query processing via voronoi
  diagram for internet of things,'' {\em Future Generation Computer Systems},
  vol.~91, pp.~382--391, 2019.

\end{thebibliography}

\begin{IEEEbiography}[{\includegraphics[width=1in,height=1.25in,clip,keepaspectratio]{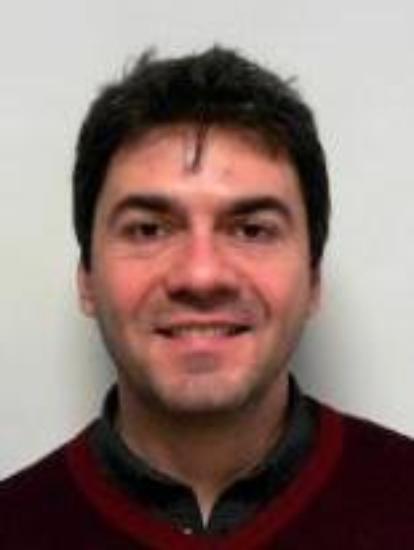}}]{Michele Albano} (M'11-SM'19)
is a tenure-track Assistant Professor for the Department of Computer Science of Aalborg University, Denmark. He got his Ph.D. in Computer Science from the University of Pisa, Italy, and his research is focused on distributed systems and embedded systems.
He has been active in more than 10 European research projects, 
and he acted as technical manager for CELTIC project Green-T; work package leader for FP7 IP ROMEO, ITEA2 CarCoDe, and ECSEL MANTIS.
Michele is editor of the Open Access book ''The MANTIS book: Cyber Physical System Based Proactive Maintenance'', and he is Editor in Chief for the Journal of Industrial Engineering and Management Science, River Publishers.
\end{IEEEbiography}

\begin{IEEEbiography}[{\includegraphics[width=1in,height=1.25in,clip,keepaspectratio]{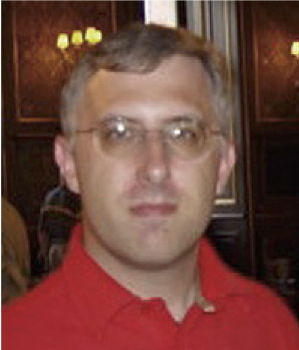}}]{Matteo Mordacchini} is a Researcher at  the Ubiquitous Internet Lab, IIT-CNR, Italy. Currently, his main research areas include Edge computing, the Internet of People paradigm, and adaptive, self-organizing distributed solutions. 
Specifically, he is investigating how models of human cognitive processes, coming from the cognitive psychology domain, can be exploited to devise autonomic and adaptive solutions for autonomous agents. Other research directions include Cloud computing and Opportunistic Networks.
Matteo has also worked in several EU projects and has served in the TPC of many international conferences and workshops.
\end{IEEEbiography}

\begin{IEEEbiography}[{\includegraphics[width=1in,height=1.25in,clip,keepaspectratio]{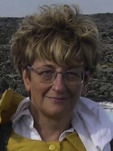}}]{Laura Ricci} received the M.Sc. degree in computer science and the Ph.D. degree from the University of Pisa, Pisa, Italy, in 1983 and 1990, respectively. She is currently an Associate Professor with the Department of Computer Science, University of Pisa. She has been involved in several research projects. She is also the Local Coordinator of the H2020 European Project Helios: A Context-aware Distributed Networking Framework. Her current research interests include distributed systems, peer-to-peer networks, cryptocurrencies and blockchains, and social network analysis. She has coauthored over 100 articles published in international journals and conference/workshop proceedings in these fields. Dr. Ricci has served as a program committee member and the chair for several conferences. She has been the Program Chair of the 19th edition of International Conference on Distributed Applications and Interoperable Systems (DAIS). She is an Organizer of the Large Scale Distributed Virtual Environments (LSDVE) Workshop held in conjunction with EUROPAR conference.
\end{IEEEbiography}

\EOD

\end{document}